\documentclass[12pt, final]{l4dc2023}

\title[Physics-informed Learning for Risk Probability Estimation]{A Generalizable Physics-informed Learning Framework for Risk Probability Estimation}
\usepackage{times}
\author{%
 \Name{Zhuoyuan Wang} \Email{zhuoyuaw@andrew.cmu.edu}\\
 \Name{Yorie Nakahira} \Email{yorie@cmu.edu}\\
 \addr Electrical and Computer Engineering Department, Carnegie Mellon University, Pittsburgh, PA, USA%
}

\usepackage{bbm}
\usepackage{times}
\usepackage{bm}
\usepackage{caption}
\usepackage{enumitem}

\newcommand{\ie}{\textit{i.e., }}
\newcommand{\eg}{\textit{e.g., }}
\newcommand{\pr}{\mathbb{P}}
\newcommand{\safe}{\mathcal{C}}
\newcommand{\D}{D}

\begin{document}

\maketitle

\begin{abstract}%
Accurate estimates of long-term risk probabilities and their gradients are critical for many stochastic safe control methods. However, computing such risk probabilities in real-time and in unseen or changing environments is challenging. Monte Carlo (MC) methods cannot accurately evaluate the probabilities and their gradients as an infinitesimal devisor can amplify the sampling noise. In this paper, we develop an efficient method to evaluate the probabilities of long-term risk and their gradients. The proposed method exploits the fact that long-term risk probability satisfies certain partial differential equations (PDEs), which characterize the neighboring relations between the probabilities, to integrate MC methods and physics-informed neural networks. We provide theoretical guarantees of the estimation error given certain choices of training configurations. Numerical results show the proposed method has better sample efficiency, generalizes well to unseen regions, and can adapt to systems with changing parameters. The proposed method can also accurately estimate the gradients of risk probabilities, which enables first- and second-order techniques on risk probabilities to be used for learning and control.
\end{abstract}

\begin{keywords}%
  Stochastic safe control; physics-informed learning; risk probability estimation.%
\end{keywords}

\vspace{0.1in}
\section{Introduction}

Safe control for stochastic systems is important yet a key challenge for deploying autonomous systems in the real world.
In the past decades, many stochastic control methods have been proposed to ensure safety of systems with noises and uncertainties, including stochastic reachabilities~\citep{prandini2008application}, chance-constrained predictive control~\citep{nakka2020chance} and \textit{etc}. Despite the huge amount of stochastic safe control methods, many of them rely on accurate estimates of long-term risk probabilities or their gradients to guarantee long-term safety~\citep{abate2008probabilistic, chapman2019risk, santoyo2021barrier, wang2021myopically}.
To get such accurate estimates is non-trivial, and here we list the challenges.
\vspace{-0.05in}
\begin{itemize}[leftmargin=.2in]
\setlength\itemsep{-0.03in}
    \item \textit{High computational complexity}. Estimating long-term risk is computationally expensive, because the possible state trajectories scale exponentially with regard to the time horizon. Besides, the values of risk probability are often small in safety-critical systems, and thus huge amounts of sample trajectories are needed to capture the rare unsafe event~\citep{janssen2013monte}. 
    
    \item \textit{Sample inefficiency.} Generalization of risk estimation to the full state space is hard to achieve for sample-based methods, since each point of interest requires one separate simulation. The sample complexity increases linearly with respect to the number of points for evaluation. Besides, most of the existing methods require re-evaluation of the risk probability for any changes of system parameters, which further degrades sample efficiency~\citep{zuev2015subset}.

    \item \textit{Lack of direct solutions}. For control affine systems, recent study~\citep{chern2021safe} suggests that the risk probability can be characterized by the solution of partial differential equations (PDEs). It provides an analytical approach to calculate the risk probability, but to solve the time-varying PDE and get the actual probability value is non-trivial. 
    
    \item \textit{Noisy estimation of probability gradients.} Estimating gradients of risk probabilities is difficult, as sampling noise is amplified by the infinitesimal divisor while taking derivatives over the estimated risk probabilities.
  
\end{itemize}
\vspace{-0.3em}
To resolve the abovementioned issues, we propose a physics-informed learning framework PIPE (Physics-Informed Probability Estimator) to estimate risk probabilities in an efficient and generalizable manner. Fig.~\ref{fig:overview diagram} shows the overview diagram of the proposed PIPE framework. The framework takes both data and physics models into consideration, and by combining the two, we achieve better sample efficiency and the ability to generalize to unseen regions in the state space and unknown parameters in the system dynamics. The use of deep neural networks enables the efficient learning of complex PDEs, and the consideration of physics models further enhances efficiency as it allows imperfect noisy data for training. The resulting framework takes only inaccurate sample data in a sparse sub-region of the state space for training, and is able to accurately predict the risk probability over the whole state space for systems with different parameters. 

\begin{figure*}
    \centering
    \includegraphics[width=0.95\textwidth]{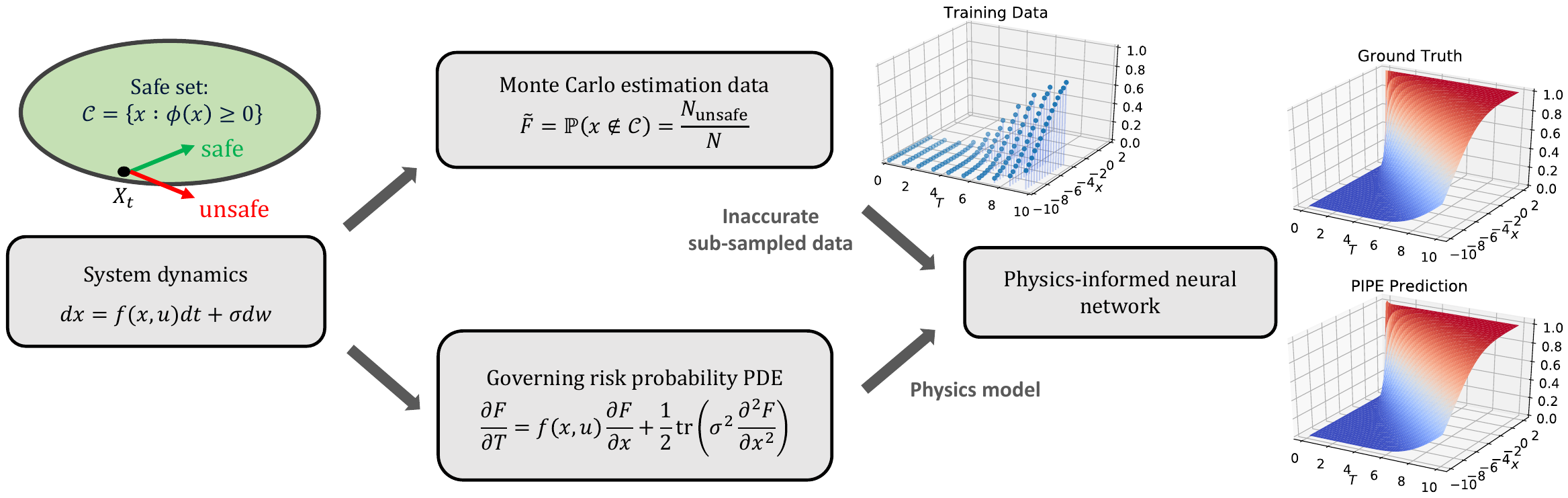}
    \caption{The overview diagram of the proposed PIPE framework. 
    The system takes form~\eqref{eq:cdc system} with safe set defined as~\eqref{eq:safe set definition}. For training data, one can acquire the empirical risk probabilities by simulating the system dynamics and calculating the ratio of unsafe trajectories over all trajectories. For the physics model, we know that the mapping between state-time pair and the risk probability satisfies a governing convection diffusion equation (Theorem~\ref{thm:cdc pde}). The PIPE framework uses physics-informed neural networks to learn the risk probability by fitting the empirical training data and using the physics model as constraints. PIPE gives more accurate and sample efficient risk probability predictions than Monte Carlo or its variants, and can generalize to unseen regions in the state space and unknown parameters in the system dynamics thanks to its integration of data and physics models. 
    % \todo{change to general nonlinear dynamics $f(x,u)$}
    }
    \label{fig:overview diagram}
\end{figure*}

\vspace{-0.3em}

\section{Related Works}
\subsection{Stochastic safe control}
% \vspace{-.1em}
Stochastic safe control is a heated topic in recent years, as safety under uncertainties and noises becomes the key challenge of many real-world autonomous systems. Stochastic reachability analysis takes the stochastic dynamics model and forward rollouts the possible trajectories to get the safety probability of any given state, and use this information to design suitable safe controllers~\citep{abate2008probabilistic, prandini2008application, patil2022upper}. 
Conditional Value-at-Risk considers the risk measure of the system and guarantees the expected risk value to always decrease conditioned on the previous states of the system, and thus guarantees safety~\citep{samuelson2018safety, singletary2022safe}.
Chance-constrained predictive control takes probabilistic safety requirements as the constraints in an optimization-based controller, and solves a minimal distance control to a nominal controller to find its safe counterpart~\citep{nakka2020chance, zhu2019chance, pfrommer2022safe}. While these methods provide theoretical guarantees on safety, all of them require accurate estimation of risk probabilities or their gradients to yield desirable performance, and to get accurate estimates itself challenging. We tackle this problem by combining physics models and data to provide accurate estimates of risk probability and its gradient.

% \vspace{-.6em}
\subsection{Rare event simulation}
% \vspace{-.3em}
Rare event simulation considers the problem of estimating the probability of rare events in the system, and is highly related to risk probability estimation because the risk probability is often small in safety-critical systems. Here, we list a few widely adopted approaches for rare event simulation. Standard MC forward runs the system dynamics multiple times to empirically estimate the risk probability by calculating the unsafe trajectory numbers over the total trajectory number~\citep{rubino2009rare}. Standard MC is easy to implement, but needs huge amounts of sample trajectories to get accurate estimation, which becomes impractical when the required accuracy is high. 
Importance sampling methods calculate the risk probability on a shifted new distribution to improve the sample efficiency, but needs good prior information on the re-sampled distribution for reasonable performance enhancement, which is hard to achieve for complex systems~\citep{cerou2012sequential, botev2013markov}. 
Subset simulation calculates the risk probability conditioned on intermediate failure events that are easier to estimate, to further improve sample efficiency~\citep{au2001estimation}. However, computational efficiency remains an issue and generalization to the entire state space is hard to achieve, as the estimation can only be conducted at a single point once~\citep{zuev2015subset}. There are no known methods that can compute the risk probability in an integrated way for the entire state space, and to do so with high sample efficiency. We address this problem by proposing a learning framework that considers both data and model to give generalizable prediction results with high sample efficiency.

% \vspace{-.6em}
\subsection{Physics-informed neural networks}
% \vspace{-.3em}
Physics-informed neural networks (PINNs) are neural networks that are trained to solve supervised learning tasks while respecting any given laws of physics described by general nonlinear partial differential equations~\citep{raissi2019physics}. PINNs take both data and the physics model of the system into account, and are able to solve the forward problem of getting PDE solutions, and the inverse problem of discovering underlying governing PDEs from data. PINNs have been widely used in power systems~\citep{misyris2020physics}, fluid mechanics~\citep{cai2022physics} and medical care~\citep{sahli2020physics}. For stochastic safe control, previous works use PINNs to solve the initial value problem of deep backward stochastic differential equations to derive an end-to-end myopic safe controller~\citep{han2018solving, pereira2021safe}. However, to the best of our knowledge there is no work that considers PINNs for risk probability estimation, especially on the full state-time space scale. In this work, we take the first step towards leveraging PINNs on the problem of risk probability estimation.

\section{Problem Formulation}
\label{sec:problem formulation}

We consider a nonlinear stochastic control system dynamics characterized by the following stochastic differential equation (SDE):
\begin{equation}
\label{eq:cdc system}
    dx = f(x, u)dt + \sigma dw,
\end{equation}
where $x \in \mathbb{R}^n$ is the state, $u \in \mathbb{R}^m$ is the control input, $w_t$ is a $n$-dimensional Wiener process starting from $w_0 = 0$ and $\sigma$ is the magnitude of the noise.
We assume that function $f$ is parameterized by some parameter $\lambda$.
Safety of the system is defined as the state staying within a safe set $\mathcal{C}$, which is the super-level set of a function $\phi(x): \mathbb{R}^n \rightarrow \mathbb{R}$, \ie 
\begin{equation}
\label{eq:safe set definition}
    \mathcal{C} = \{x \mid \phi(x) \geq 0\}.
\end{equation}
This definition of safety can characterize a large variety of practical safety requirements~\citep{prajna2007framework, ames2019control}. For the stochastic system~\eqref{eq:cdc system}, since safety can only be guaranteed in the sense of probability, we consider the long-term safety probability $F_s$ and recovery probability $F_r$ of the system defined as below. 
% \todo{argue we consider safety and recover for risk in the rest of the paper, use $\Psi$ to denote both}
\begin{definition}[Safety probability]
\label{def:safety probability}
Starting from initial state $x_0 = x \in \safe$, the safety probability $F^s$ of system~\eqref{eq:cdc system} for outlook time horizon $T$ is defined as the probability of state $x_t$ staying in the safe set $\safe$ over the time interval $[0,T]$, i.e., 
\begin{equation}
    F^s(x,T) = \pr(x_t \in \safe, \forall t\in [0,T] \mid x_0 =x).
\end{equation}
\end{definition}
\begin{definition}[Recovery probability]
\label{def:recovery probability}
Starting from initial state $x_0 = x \notin \safe$, the recovery probability $F^r$ of system~\eqref{eq:cdc system} for outlook time horizon $T$ is defined as the probability of state $x_t$ to get back to the safe set during the time interval $[0,T]$, i.e.,
\begin{equation}
    F^r(x,T) = \pr(\exists t \in [0,T], x_t \in \safe \mid x_0 =x).
\end{equation}
\end{definition}
Both safety probability and recovery probability are specific realizations of risk probability, depending on whether the initial state is safe or not and whether the safe or unsafe events are of interest in the system (they are complementary). In the rest of the paper, we will denote risk probability as $F$, and it refers to safety probability or recovery probability depending on the initial state of the system is within the safe set or not.

\begin{remark}
The risk probability defined above characterizes the long-term behaviour of the stochastic system. In many literature, people may consider Gaussian approximation to measure the risk at each time step~\citep{liu2015safe, akametalu2014reachability}. However, if we know the probability of risk at each time step being $\epsilon$, then the probability bound of risk for $k$ time steps will be $1 - (1-\epsilon)^k \rightarrow 1$. This value is very conservative because the derivation does not capture the dynamic relationships between time steps. 
Previous studies also consider approximations of long-term safety, \eg supermartingale~\citep{prajna2007framework}, Chebyshev's inequality~\citep{boucheron2013concentration, farokhi2021safe} and Boole's inequality~\citep{li2021analytical}. Those approximations allows efficient calculations but can be conservative as well.
\end{remark}

The value of the risk probability $F$ over the state space is crucial for safe control of these systems, but in practice it is hard to obtain all risk probability information accurately, \eg the failure probability of manufacturing robot arm is very small and thus hard to estimate~\citep{lasota2014toward}.
% \todo{replace with safe robotic review paper} 
In this paper, the goal is to accurately estimate the risk probability and its gradient over the entire state space, and to achieve adaptation in changing system dynamics. Specifically, we want to find the mapping between the state-time pair $(x,T)$ and the risk probability $F$ over the entire state-time space $\mathbb{R}^n \times \mathbb{R}$ for any system parameter $\lambda$, and to estimate the gradient of risk probability $F$ with regard to the state $x$.
We list four specific goals for the problem, generalization to unseen regions in the state-time space, efficient estimation with fixed number of sample data, adaptation on changing system parameters, and accurate estimation of probability gradients.

\section{Proposed Method}
\begin{figure*}[t]
    \centering
    \includegraphics[width=0.85\textwidth]{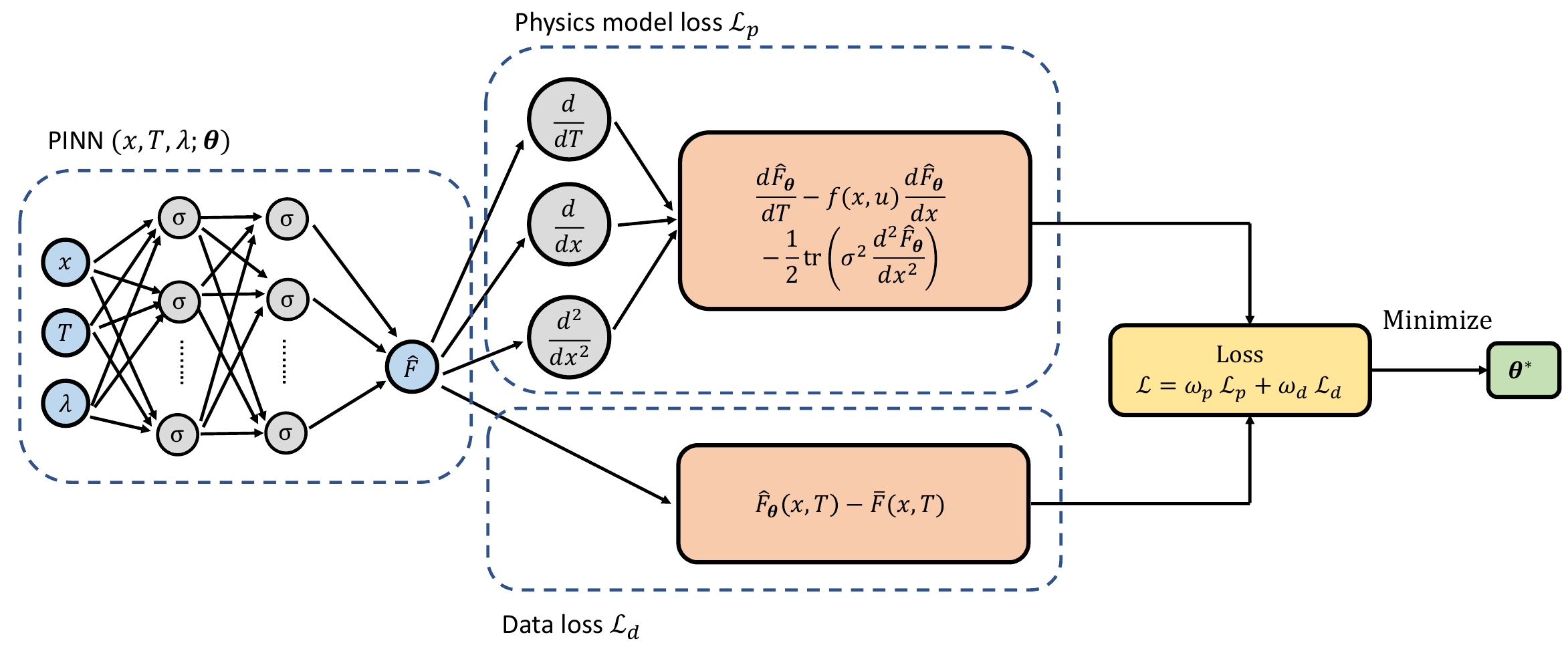}
    \caption{The training scheme of the physics-informed neural network (PINN) in PIPE. 
    % The PINN takes state $x$, time horizon $T$ and the system parameter $\lambda$ as input, and outputs the prediction of risk probability $\hat{F}_{\boldsymbol\theta}$. The PINN is parameterized by $\boldsymbol{\theta}$. The loss function is the weighted sum of the physics model loss $\mathcal{L}_p$ and data loss $\mathcal{L}_d$. The physics model loss calculates the satisfaction of the governing PDE on sampled data points in the state-time space. The partial derivatives and Hessian of the risk probability come naturally from the automatic differentiation in deep learning frameworks such as PyTorch~\citep{paszke2019pytorch} and TensorFlow~\citep{abadi2016tensorflow}. At the end, we use an optimizer to minimize the loss function and find the optimal parameters in the PINN. 
    % In testing phase, the PINN only needs a forward pass of the network to yield prediction of risk probabilities.
    }
    \label{fig:PINN diagram}
\end{figure*}

% \todo{put this as introductory paragraph at the beginning of the proposed method section. Done}

In this section, we propose a model-based data-driven approach, PIPE, to efficiently estimate risk probabilities. We first introduce a PDE whose solution will characterize the risk probability. Then we combine MC data and the PDE model to form a physics-informed neural network to learn the risk probability. The PIPE formulation enjoys advantages from both PDE and MC, and achieves better performance and efficiency than any of the method alone.

\begin{theorem}~\citep{chern2021safe}
\label{thm:cdc pde}
For a stochastic control affine system in the form of~\eqref{eq:cdc system}, the risk probability $F$ (for both safety probability $F^s$ defined in Definition~\ref{def:safety probability} and recovery probability $F^r$ defined in Definition~\ref{def:recovery probability}) is characterized by the following convection diffusion equation,
\begin{equation}
\label{eq:cdc pde}
\begin{aligned}
     W_F(x,T) &:= \frac{\partial F}{\partial T}(x, T) - f(x, u) \frac{\partial F}{\partial x}(x, T) -\frac{1}{2} \operatorname{tr}\left(\sigma^2 \frac{\partial^{2} F}{\partial x^{2}}(x, T) \right) = 0,
\end{aligned}
\end{equation}
with initial condition $F(x,0) = \mathbbm{1}(x \in \safe)$.
For safety probability, the boundary condition is $F^s(x,T) = 0 ,\; x \notin \safe$.
For recovery probability, the boundary condition is $F^r(x,T) = 1 ,\; x \in \partial \safe$.
\end{theorem}

% \begin{equation}
% \label{eq:cdc recovery pde}
% \begin{aligned}
%      W(x,T) &= \frac{\partial F}{\partial t}(x, T) - (f(x)+g(x)u) \frac{\partial F}{\partial x}(x, T) \\
%     & \qquad +\frac{1}{2} \operatorname{tr}\left(\sigma^2 \frac{\partial^{2} F}{\partial x^{2}}(x, T) \right) = 0,
% \end{aligned}
% \end{equation}
% with initial and boundary conditions
% \begin{equation}
% \begin{aligned}
% \label{eq:cdc icbc}
% F(x,T) &= 0 \quad x \in \partial \safe, \\
% F(x,0) &= \mathbbm{1}(x \notin \safe).
% \end{aligned}
% \end{equation}
% \todo{add the relation between $W$ and $F$.}
Theorem~\ref{thm:cdc pde} states that the risk probability of a control system can be analytically expressed as a PDE. The PDE consists of a convection term and a diffusion term. The convection term characterizes how the risk probability changes as the system evolves under its deterministic part of the dynamics. The diffusion term characterizes the effect of stochastic noises on the risk probability value as the noise term in the dynamics diffuses with time.
The initial condition says when the outlook time $T$ is $0$, the risk probability value is the indicator function of whether the state is within the safe set. The boundary condition says that on the boundary of safe set $\mathcal{C}$, the risk probability $F$ is $0$ if we consider safety probability, and is $1$ if we consider recovery probability. 
We use $W_F(x,T)$ to denote the function value that the PDE risk probability $F$ should satisfy, to better define the loss function in the learning framework later.

% \todo{the boundary condition is different, depending on if it is safe or recovery probability. can you do either one of the following? 1. define one probability (from which both safe probability or recovery probability can be computed.) and state the current boundary condition. 2. define two F (safey vs recovery), and state two PDEs.}
% \todo{maybe its more rigorous to put the relation between PDE in theorem statement?}
While the PDE provides a way to get the actual risk probability of the system, to solve a PDE using numerical techniques is not easy in general, especially when the coefficients are time varying as in the case of~\eqref{eq:cdc pde}.
MC methods provide another way to solve this problem. Assume the dynamics of the system is given, one can simulate the system for an initial condition multiple times to get an empirical estimate of the risk probability by calculating the ratio of unsafe trajectories over all trajectories. However, 
MC requires huge number of trajectories to get accurate estimation, and the evaluation of the risk probability can only be conducted at a single point at a time.

To leverage the advantages of PDE and MC and to overcome their drawbacks, we propose to use physics-informed neural networks (PINNs) to learn the mapping from the state and time horizon to the risk probability value $F$. Fig.~\ref{fig:PINN diagram} shows the architecture of the PINN. The PINN takes the state-time pair $(x,T)$ and the system parameter $\lambda$ as the input, and outputs the risk probability prediction $\hat{F}$, the state and time derivatives $ \frac{\partial \hat F}{\partial x}$ and $ \frac{\partial \hat F}{\partial T}$, and the Hessian $ \frac{\partial^2 \hat F}{\partial x^2}$, which come naturally from the automatic differentiation in deep learning frameworks such as PyTorch~\citep{paszke2019pytorch} and TensorFlow~\citep{abadi2016tensorflow}. Unlike standard PINN, we add the system parameter $\lambda$ as an input to achieve adaptations on varying system parameters.
Assume the PINN is parameterized by $\boldsymbol\theta$, the loss function is defined as
\begin{equation}
\label{eq:PINN overall loss function}
    \mathcal{L}(\boldsymbol\theta) = \omega_p \mathcal{L}_p(\boldsymbol\theta) + \omega_d \mathcal{L}_d(\boldsymbol\theta),
\end{equation}
where
\begin{equation}
\label{eq:PINN loss functions}
\begin{aligned}
    \mathcal{L}_p(\boldsymbol\theta) & = \frac{1}{|\mathcal{P}|} \sum_{(x,T) \in \mathcal{P}} \|W_{\hat{F}_{\boldsymbol\theta}}(x,T)\|_2^2, \\
    \mathcal{L}_d(\boldsymbol\theta) & = \frac{1}{|\mathcal{D}|} \sum_{(x,T) \in \mathcal{D}} \|\hat{F}_{\boldsymbol\theta}(x,T) - \bar{F}(x,T)\|_2^2.
\end{aligned}
\end{equation}
Here, $\bar{F}$ is the training data, $\hat{F}_{\boldsymbol\theta}$ is the prediction from the PINN, $\mathcal{P}$ and $\mathcal{D}$ are the training point sets for the physics model and external data, respectively. The loss function $\mathcal{L}$ consists of two parts, the physics model loss $\mathcal{L}_p$ and data loss $\mathcal{L}_d$. The physics model loss $\mathcal{L}_p$ measures the satisfaction of the PDE constraints for the learned output. It calculates the actual PDE equation value $W_{\hat{F}_{\boldsymbol\theta}}$, which is supposed to be $0$, and use its 2-norm as the loss. The data loss $\mathcal{L}_d$ measures the accuracy of the prediction of PINN on the training data. It calculates the mean square error between the PINN prediction and the training data point as the loss. The overall loss function $\mathcal{L}$ is the weighted sum of the physics model loss and data loss with weighting coefficients $\omega_p$ and $\omega_d$.

The resulting PIPE framework combines MC data and the governing PDE into a PINN to learn the risk probability. The advantages of the PIPE framework include fast inference at test time, accurate estimation, and ability to generalize from the combination of data and model.

\vspace{-.05in}
\section{Performance Analysis}
In this section, we provide performance analysis of PIPE. We first show that for standard neural networks (NNs) without physics model constraints, it is fundamentally difficult to estimate the risk probability of a longer time horizon than those generated from sampled trajectories. We then show that with the PINN, we are able to estimate the risk probability at any state for any time horizon with bounded error.
Let $\Omega$ be the state space, $\tau = [0,T_H]$ be the time domain, $\Sigma=(\partial \Omega \times[0, T_H]) \cup(\Omega \times\{0\})$ be the boundary of the space-time domain.
% Let $\Tilde{u}(x,T), (x,T) \in \Sigma$ be the boundary condition of the PDE~\eqref{eq:cdc pde}. 
Denote $\D:=\Omega \times \tau$ for notation simplicity and denote $\Bar{\D}$ be the interior of $\D$.

% In this section, we provide performance analysis for the PIPE framwork in terms of error bounds. In practice, one can choose different sets $\mathcal{P}$ and $\mathcal{D}$ to impose PDE constraints and training data loss. Here we present results on some typical choices of $\mathcal{P}$ and $\mathcal{D}$.  
% Let $\Omega$ be the state space, $\tau = [0,T_H]$ be the time domain, $\Sigma=(\partial \Omega \times[0, T_H]) \cup(\Omega \times\{0\})$ be the boundary of the space-time domain.
% Let $\Tilde{u}(x,T), (x,T) \in \Sigma$ be the boundary condition of the PDE~\eqref{eq:cdc pde}. We present the following theorems to show the performance guarantees when we choose $\mathcal{P} = \Omega \times \tau$ and $\mathcal{D}$, $\mathcal{P}$ being sub-set of $\Omega \times \tau$. Denote $\D:=\Omega \times \tau$ for notation simplicity and denote $\Bar{\D}$ be the interior of $\D$.

\setlength{\belowdisplayskip}{4pt} \setlength{\belowdisplayshortskip}{4pt}
\setlength{\abovedisplayskip}{4pt} \setlength{\abovedisplayshortskip}{4pt}

\vspace{-.05in}
\begin{corollary}
\label{cor:NN_worst_case_bound}
Suppose that $\D \in \mathbb{R}^{d+1}$ is a bounded domain, $u \in C^0(\bar{\D}) \cap C^2(\D)$ is the solution to the PDE of interest, and $\Tilde{u}(x,T), (x,T) \in \Sigma$ is the boundary condition. 
Let $\Sigma_s$ be a strict sub-region in $\Sigma$, and $\D_s$ be a strict sub-region in $\D$.
Consider a neural network ${F}_{\boldsymbol\theta}$ that is parameterized by 
$\boldsymbol\theta$ and has sufficient representation capabilities. For $\forall M > 0$, there can exist $\Bar{\boldsymbol\theta}$ that satisfies both of the following conditions simultaneously: 
\vspace{-0.02in}
\begin{enumerate}
\setlength\itemsep{-0.02in}
    \item 
    $\sup _{(x, T) \in  \Sigma_s}|{F}_{\Bar{\boldsymbol\theta}}(x, T)-\Tilde{u}(x, T)|<\delta_1$
    \item
    $\sup _{(x, T) \in \D_s}|{F}_{\Bar{\boldsymbol\theta}}(x, T)-{u}(x, T) |<\delta_2$
    % \item 
    % ${F}_{\boldsymbol\theta} \in C^0(\bar{\D}) \cap C^2(\D)$
\end{enumerate}
\vspace{-0.02in}
and
%\vspace{-0.1in}
\begin{equation}
\label{eq:PINN_worst_case_bound}
    \sup _{(x, T) \in \D}\left|{F}_{\Bar{\boldsymbol\theta}}(x, T)-u(x, T)\right| \geq M.
\end{equation}   
\end{corollary}

\begin{theorem}
\label{thm:full_pde_constraint}
Suppose that $\D \in \mathbb{R}^{d+1}$ is a bounded domain, $u \in C^0(\bar{\D}) \cap C^2(\D)$ is the solution to the PDE of interest, and $\Tilde{u}(x,T), (x,T) \in \Sigma$ is the boundary condition. Let ${F}_{\boldsymbol\theta}$ denote a PINN parameterized by 
$\boldsymbol\theta$. If the following conditions holds:
% \vspace{-0.02in}
\begin{enumerate}
\setlength\itemsep{-0.02in}
    \item $\mathbb{E}_{\mathbf{Y}}\left[|{F}_{\boldsymbol\theta}(\mathbf{Y})-\Tilde{u}(\mathbf{Y})|\right]<\delta_1$, where $\mathbf{Y}$ is uniformly sampled from $\Sigma$
    \item $\mathbb{E}_{\mathbf{X}}\left[|W_{{F}_{\boldsymbol\theta}}(\mathbf{X})|\right]<\delta_2$, where $\mathbf{X}$ is uniformly sampled from $\D$
    \item 
    $F_{\boldsymbol\theta}$, $W_{{F}_{\boldsymbol\theta}}$, $u$ are $\frac{l}{2}$ Lipshitz continuous on $\D$.
\end{enumerate} 
% \vspace{-0.02in}
Then the error of ${F}_{\boldsymbol\theta}$ over $\D$ is bounded by
% \vspace{-0.06in}
% \begin{equation}
% \label{eq:pde_constraint_bound}
%     \sup _{(x,T) \in \D}\left|{F}_{\boldsymbol\theta}(x, T)-u(x, T)\right| \leq C_1 \delta_1 + C_2 \frac{\delta_2}{\lambda}
% \end{equation}
\begin{equation}
\label{eq:pde_constraint_bound}
    \sup _{(x,T) \in \D}\left|{F}_{\boldsymbol\theta}(x, T)-u(x, T)\right| \leq \tilde \delta_1 + C \frac{\tilde \delta_2}{\sigma^2}
\end{equation}
% \vspace{-0.08in}
% where $C_1$, $C_2$ are constants depending on $\D$, $\Sigma$, and $W$.
where $C$ is a constant depending on $\D$, $\Sigma$ and $W$, and
\begin{equation}
\begin{aligned}
    \tilde \delta_1 & = \max \left\{\frac{2 \delta_1 |\Sigma| }{R_{\Sigma}|\Sigma|}, 2 l \cdot\left(\frac{\delta_1 |\Sigma| \cdot \Gamma(d / 2+1)}{l R_{\Sigma} \cdot \pi^{d / 2}}\right)^{\frac{1}{d+1}}\right\}, \\
    \tilde \delta_2 & = \max \left\{\frac{2\delta_2 |\D|}{R_{\D}|\D|}, 2 l \cdot\left(\frac{\delta_2 |\D| \cdot \Gamma((d+1) / 2+1)}{l R_{\D} \cdot \pi^{(d+1) / 2}}\right)^{\frac{1}{d+2}}\right\},
\end{aligned}
\end{equation}
with $R_{(\cdot)}$ being the regularity of $(\cdot)$, $\|(\cdot)\|$ is the Lebesgue measure
of a set $(\cdot)$ and $\Gamma$ is the Gamma function.
\end{theorem}

% \begin{proof} See Appendix.
% \end{proof}

% See extended version of this paper for proofs~\citep{wang2023generalizable}. 
See Appendix~\ref{appdx:proofs} for the proofs.
Corollary~\ref{cor:NN_worst_case_bound} says that standard NN can give arbitrarily inaccurate prediction due to insufficient physical constraints. This explains why risk estimation problems cannot be handled solely on fitting sampled data. Theorem~\ref{thm:full_pde_constraint} says that when the PDE constraint is imposed on the full space-time domain, the prediction of the PINN has bounded error.
% , and the error scales linearly with the training loss. 

% See Appendix~\ref{appdx:proofs} for the proofs.
% Theorem~\ref{thm:PINN_worst_case_bound} says that when we only impose data and PDE constraints on a sub-region in the space-time domain, the learned PINN can perform arbitrarily inaccurate due to insufficient physical constraints. One special case of this setting is $\mathcal{P} = \emptyset$, which is the typical configuration for standard neural networks without physical constraints. This explains why risk estimation problems cannot be handled solely on fitting sampled data.
% We also point out that Theorem~\ref{thm:PINN_worst_case_bound} is a worst-case result. In practice, one may observe that the PINN performs well on the full space-time domain even with constraints imposed on sub-region. 
% Theorem~\ref{thm:full_pde_constraint} says that when the PDE constraint is imposed on the full space-time domain, the prediction of the PINN has bounded error, and the error scales linearly with the training loss. 

% \mathcal{L}_p(\boldsymbol\theta) & = \frac{1}{\mathcal{P}} \sum_{(x,T) \in \mathcal{P}} \|W_{\hat{F}_{\boldsymbol\theta}}(x,T)\|_2^2, \\
%     \mathcal{L}_d(\boldsymbol\theta) & = \frac{1}{\mathcal{D}} \sum_{(x,T) \in \mathcal{D}} \|\hat{F}_{\boldsymbol\theta}(x,T) - \bar{F}(x,T)\|_2^2.

\vspace{-0.1in}
\section{Experiments}
\label{sec:experiments}
We conduct four experiments to illustrate the efficacy of the proposed method. The system dynamics of interest is~\eqref{eq:cdc system}
with $x \in \mathbb{R}$, $f(x) = \lambda \: dt$, $g(x)=0$ and $\sigma=1$. The system dynamics become
\vspace{-0.01in}
\begin{equation}
\label{eq:experiment system dynamics}
    dx = \lambda \: dt + dw.
\end{equation}
% \vspace{-0.03in}
The safe set is defined as~\eqref{eq:safe set definition} with $\phi(x) = x-2$. The state-time region of interest is $\Omega \times \tau = [-10,2] \times [0,10]$.
For risk probability, we consider the recovery probability of the system from initial state $x_0 \notin \mathcal{C}$ outside the safe set. Specifically, the risk probability $F$ is characterized by the solution of the following convection diffusion equation
\begin{equation}
\label{eq:cdc recovery pde}
\begin{aligned}
    \frac{\partial F}{\partial T}(x, T) 
    & = \lambda \frac{\partial F}{\partial x}(x, T) + \frac{1}{2} \operatorname{tr}\left( \frac{\partial^{2} F}{\partial x^{2}}(x, T) \right),
\end{aligned}
\end{equation}
with initial condition $F(x,0) = \mathbbm{1}(x \geq 2)$ and boundary condition $F(2,T) = 1$.
% \begin{equation}
% \begin{aligned}
% \label{eq:cdc icbc}
% F(x,T) & = 1, \; x = 2, \\
% F(x,0) & = \mathbbm{1}(x \geq 2).
% \end{aligned}
% \end{equation}
% The boundary condition says on the boundary of safe set $\mathcal{C}$, the safety probability $F$ is always $1$. The initial condition says, when the outlook time $t$ is $0$, the safety probability value is the indicator function of whether the state is within the safe set.
We choose this system because we have the analytical solution of~\eqref{eq:cdc recovery pde} as ground truth for comparison, as given by 
% $F(x,T) =(2-x) \operatorname{exp}\left\{\frac{-((2-x)- \lambda T)^{2}}{2 T}\right\} / \sqrt{2 \pi T^{3}}$.
$F(x,T) = \int_0^T \frac{(2-x)}{\sqrt{2 \pi t^{3}}} \exp \left(-\frac{\left((2-x)-\lambda t\right)^2}{2 t}\right) dt$.
% where $\text{erf}(x) = \frac{2}{\sqrt{\pi}} \int_{0}^{x} e^{-t^2} \, dt$ is the error function associated with Gaussian distributions.
% \begin{equation}
% \label{eq:analytical solution}
%     F(x,T) =\frac{(2-x) \operatorname{exp}\left\{\frac{-((2-x)- \lambda T)^{2}}{2 T}\right\}}{\sqrt{2 \pi T^{3}}}.
% \end{equation}
The empirical data of the risk probability is acquired by running MC with the system dynamics~\eqref{eq:cdc system} with initial state $x=x_0$ multiple times, and calculate the number of trajectories where the state recovers to the safe set during the time horizon $[0,T]$ over the full trajectory number, \ie $\bar{F}(x,T) = \pr(\exists t \in [0,T], x_t \in \mathcal{C} \mid x_0 = x) = \frac{N_\text{recovery}}{N},$
% \begin{equation}
%     \bar{F}(x,T) = \pr(\exists t \in [0,T], x_t \in \mathcal{C} \mid x_0 = x) = \frac{N_\text{recovery}}{N},
% \end{equation}
where $N$ is the number of sample trajectories and is a tunable parameter that affects the accuracy of the estimated risk probability. Specifically, larger $N$ gives more accurate estimation.

In all experiments, we use PINN with 3 hidden layers and 32 neurons per layer. The activation function is chosen as hyperbolic tangent function ($\tanh$). We use Adam optimizer~\citep{kingma2014adam} for training with initial learning rate set as $0.001$. The PINN parameters $\boldsymbol\theta$ is initialized via Glorot uniform initialization. The weights in the loss function~\eqref{eq:PINN overall loss function} are set to be $\omega_p = \omega_d = 1$. We train the PINN for 60000 epochs in all experiments. The simulation is constructed based on the DeepXDE framework~\citep{lu2021deepxde}. 
Details about the experiments, simulation results on other systems, and applications to stochastic safe control can be found in the appendix.
% Experiment details, simulation results on other systems, and applications to stochastic safe control can be found in the extended version of this paper~\citep{wang2023generalizable}.
Codes are available at: \href{https://github.com/jacobwang925/PIPE-L4DC}{https://github.com/jacobwang925/PIPE-L4DC}.

\begin{figure}
    \centering
    \includegraphics[width=0.9\textwidth]{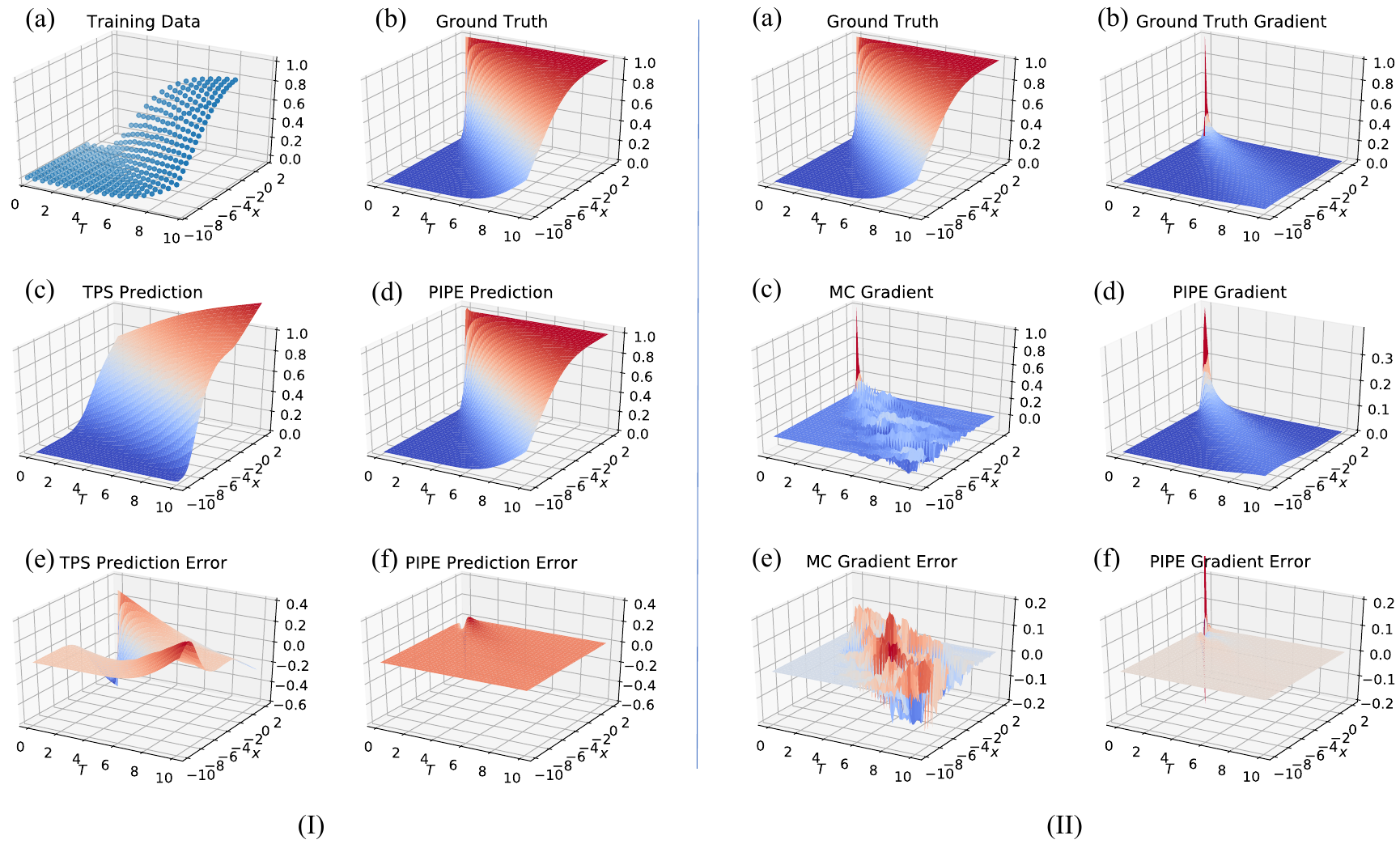}
    \caption{(I) Settings and results of the risk probability generalization task. PIPE and TPS fitting are compared. The average absolute error of prediction is $9.2 \times 10^{-2}$ for TPS, and $0.3 \times 10^{-2}$ for PIPE. (II) Gradient of the risk probability prediction of PIPE and MC. The average absolute error of gradient prediction is $2.78 \times 10^{-2}$ for MC, and $0.06 \times 10^{-2}$ for PIPE.}
    \label{fig:generalization}
\end{figure}

% \todo{This paper specifies the computing infrastructure used for running experiments (hardware and software), including GPU/CPU models; amount of memory; operating system; names and versions of relevant software libraries and frameworks.}

\subsection{Generalization to unseen regions}
\label{sec:generalization}
In this experiment, we test the generalization ability of PIPE to unseen regions of the state-time space. We consider system~\eqref{eq:experiment system dynamics} with $\lambda=1$. We train the PINN with data only on the sub-region of the state-time space $\Omega \times \tau = [-10,-2] \times [0,10]$, but test the trained PINN on the full state-time region $\Omega \times \tau = [-10,2] \times [0,10]$.
% Fig.~\ref{fig:generalization} shows the results.
The training data is acquired through MC with sample trajectory number $N = 1000$, and is down-sampled to $dx = 0.4$ and $dT = 0.5$. For comparison, we use thin plate spline (TPS) fitting on the training data to infer the risk probability on the whole state space. We also examined other fitting methods such as cubic spline and polynomial fitting, but TPS performs the best over all fitting strategies. Fig.~\ref{fig:generalization} (I) visualizes the training data samples and shows the results. The spline fitting does not include any physical model constraint, thus fails to generalize to unseen regions in the state space. On the contrary, PIPE can infer the risk probability value very accurately on the whole state space due to its combination of data and the physics model. 

\vspace{-0.05in}
\subsection{Efficient estimation of risk probability}
\label{sec:estimation}
In this experiment, we show that PIPE will give more efficient estimations of risk probability in terms of accuracy and sample number compared to MC and its variants. We consider system~\eqref{eq:experiment system dynamics} with $\lambda=1$. The training data is sampled on the state-time space $\Omega \times \tau = [-10,2] \times [0,10]$ with $dx = 0.2$ and $dT=0.1$.
We compare the risk probability estimation error of PIPE and MC, on two regions in the state-time space:
\vspace{-0.06in}
\begin{enumerate}
\setlength\itemsep{-0.01in}
    \item Normal event region: $\Omega \times \tau = [-6,-2] \times [4,6]$, where the average probability is $0.412$.
    \item Rare event region: $\Omega \times \tau = [-2,0] \times [8,10]$, where the average probability is $0.985$.
\end{enumerate}
\vspace{-0.06in}
For fairer comparison, we use a uniform filter of kernel size $3$ on the MC data to smooth out the noise, as the main cause of inaccuracy of MC estimation is sampling noise. Fig.~\ref{fig:PINN MC log plot} shows the percentage errors of risk probability inference under different MC sample numbers $N$.
As the sample number goes up, prediction errors for all three approaches decrease. The denoised MC has lower error compared to standard MC as a result of denoising, and their errors tend to converge since the sampling noise contributes less to the error as the sample number increases. On both rare events and normal events, PIPE yields more accurate estimation than MC and denoised MC across all sample numbers. This indicates that PIPE has better sample efficiency than MC and its variants, as it requires less sample data to achieve the same prediction accuracy.
This desired feature of PIPE is due to the fact that it incorporates model knowledge into the MC data to further enhance its accuracy by taking the physics-informed neighboring relationships of the data into consideration.
\begin{figure}
    \centering
    \includegraphics[width=15cm]{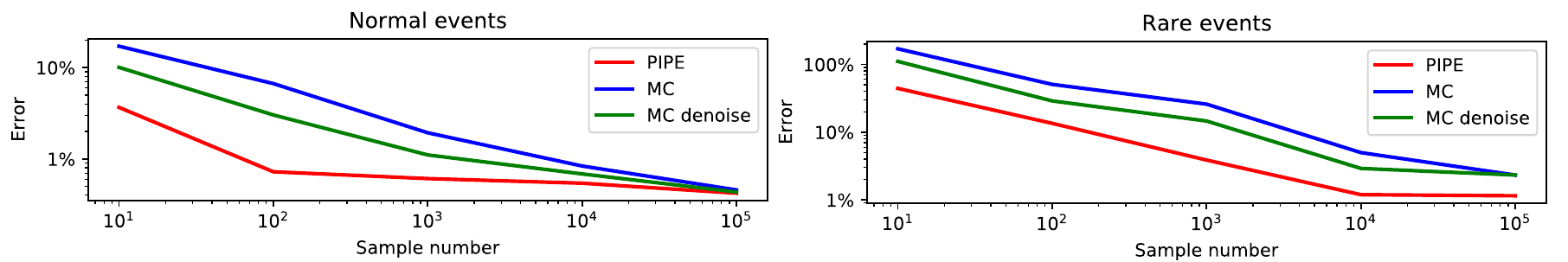}
    \caption{Percentage error of risk probability estimation for different MC sample numbers for rare events and normal events. PIPE, MC and denoised MC with uniform kernel filtering are compared. Both error and sample number are in log scale.}
    \label{fig:PINN MC log plot}
\end{figure}
\begin{figure}
    \centering
    \includegraphics[width=11cm]{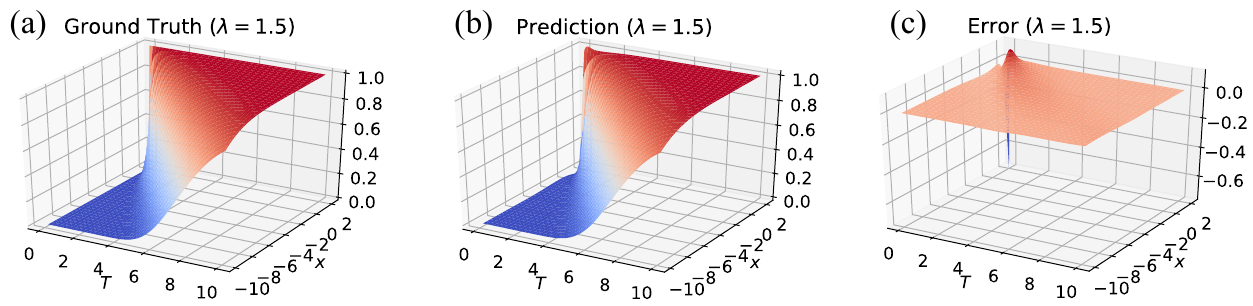}
    \caption{Risk probability prediction of PIPE on unseen system parameters. The average absolute error of the prediction is $0.70 \times 10^{-2}$.}
    \label{fig:varying parameter}
\end{figure}

\vspace{-0.05in}
\subsection{Adaptation on changing system parameters}
\label{sec:adaptation}
In this experiment, we show that PIPE will allow generalization to uncertain parameters of the system. We consider system~\eqref{eq:experiment system dynamics} with varying $\lambda\in [0,2]$. We use MC data with sample number $N=10000$ for a fixed set of $\lambda_\text{train} = [0.1,0.5,0.8,1]$ for training, and test PIPE after training on $\lambda_\text{test} = [0.3, 0.7, 1.2, 1.5, 2]$. We only present $\lambda_\text{test} = 1.5$ due to space limitation, but similar results on different $\lambda_\text{test}$ can be found at the project webpage. Fig.~\ref{fig:varying parameter} shows the results.
We can see that PIPE is able to predict risk probability for systems with unseen and even out of distribution parameters during training. In the prediction error plot, the only region that has larger prediction error is at $T=0$ and $x \in \partial \mathcal{C}$ on the boundary of the safe set. This is because the risk probability at this point is not well defined (it can be either $0$ or $1$), and this point will not be considered in a control scenario as we are always interested in long-term safety where $T \gg 0$. This adaptation feature of the PIPE framework indicates its potential use on stochastic safe control with uncertain system parameters, and it also opens the door for physics-informed learning on a family of PDEs. In general, PDEs with different parameters can have qualitatively different behaviors, so is hard to generalize. The control theory model allows us to have a sense when the PDEs are qualitatively similar with different parameters, and thus allows generalization within the qualitatively similar cases. 

% 0.002279884134690839
% 0.007021471731404468
\vspace{-0.06in}
\subsection{Estimating the gradient of risk probability}
\label{sec:gradient}
In this experiment, we show that PIPE is able to generate accurate gradient predictions of risk probabilities. We consider system~\eqref{eq:experiment system dynamics} with $\lambda=1$. Similar to the generalization task, we train the PINN with MC data of $N=1000$ on the sub-region $\Omega \times \tau = [-10,-2] \times [0,10]$ and test the trained PINN on the full state-time region $\Omega \times \tau = [-10,2] \times [0,10]$. We then take the finite difference of the risk probability with regard to the state $x$ to calculate its gradient, for ground truth $F$, MC estimation $\bar{F}$ and PIPE prediction $\hat{F}_{\boldsymbol\theta}$. Fig.~\ref{fig:generalization} (II) shows the results. It can be seen that PIPE gives much more accurate estimation of the risk probability gradient than MC, and this is due to the fact that PIPE incorporates physics model information inside the training process. It is also interesting that PIPE does not use any governing laws of the risk probability gradient during training, and by considering the risk probability PDE alone, it can provide very accurate estimation of the gradient. The results indicates that PIPE can enable the usage a lot of first- and higher-order stochastic safe control methods online, by providing accurate and fast estimation of the risk probability gradients.

\vspace{-0.09in}
\section{Conclusion}
In this paper, we proposed a generalizable physics-informed learning framework, PIPE, to estimate the risk probability in stochastic safe control systems. The PIPE framework combines data from Monte Carlo and the underlying governing PDE of the risk probability, to accurately learn the risk probability as well as its gradient. PIPE has better sample efficiencies compared to MC, and is able to generalize to unseen regions in the state space beyond training. The resulting PIPE framework is also robust to uncertain parameters in the system dynamics, and can infer risk probability values of a class of systems with training data only from a fixed number of systems. The proposed PIPE framework provides key foundations for first- and higher-order methods for stochastic safe control, and opens the door for robust physics-informed learning for generic PDEs. Future work includes applications to high-dimensional and real-world systems.

% Acknowledgments---Will not appear in anonymized version
% \acks{We thank a bunch of people.}

% \section*{Acknowledgement}
% This paper is supported in part by the Department of the Navy, Office of Naval Research, grant number N00014-23-1-2252. Any opinions, findings, and conclusions or recommendations expressed in this material are those of the author(s) and do not necessarily reflect the views of the Office of Naval Research.

\bibliography{l4dc_2023/main}

\newpage

\appendix
% \section*{Appendix}
% \label{sec:appendix}

\setlength{\belowdisplayskip}{7pt} \setlength{\belowdisplayshortskip}{7pt}
\setlength{\abovedisplayskip}{7pt} \setlength{\abovedisplayshortskip}{7pt}

\section{Proof of Theorems}
\label{appdx:proofs}

In this section we provide proofs of the corollary and theorems listed in this paper. 

\begin{proof}(Corollary~\ref{cor:NN_worst_case_bound})
We know that $u$ is the solution to the PDE of interest. 
% Let $u^* = u$ be the solution on $\D_s$. 
We can construct $\Bar{u}$ such that
\begin{equation}
    \bar{u} = \begin{cases}
    u, \quad (x,T) \in \D_s \\
    u + \frac{d}{d_{\text{max}}}(M+\delta), \quad (x,T) \in \Omega \times \tau \backslash \D_s
    \end{cases}
\end{equation}
where $d$ characterizes the distance between $(x,t) \in \Omega \backslash \D_s$ and set $\D_s$, $d_{\text{max}} = \max_{(x,t) \in \Omega \backslash \D_s} d$ is the maximum distance, and $0 < \delta < \min\{\delta_1, \delta_2\}$ is a constant. Then we have
\begin{equation}
    \sup_{(x, T) \in \Omega \times \tau} \left|\bar{u}(x, t)-u(x, t)\right| = M + \delta
\end{equation}
Since we assume the neural network has sufficient representation power, by universal approximation theorem~\citep{cybenko1989approximation}, for $\delta$ given above, there exist $\Bar{\boldsymbol\theta}$ such that
\begin{equation}
    \sup _{(x, T) \in \Omega \times \tau}\left|F_{\Bar{\boldsymbol\theta}}(x, t)-\bar{u}(x, t)\right| \leq \delta.
\end{equation}
Then we have $\Bar{\boldsymbol\theta}$ satisfies 
\begin{equation}
\begin{aligned}
\sup _{(x, T) \in  \Sigma_s}|{F}_{\Bar{\boldsymbol\theta}}(x, T)-\Tilde{u}(x, T)|<\delta_1, \\
\sup _{(x, T) \in \D_s}|{F}_{\Bar{\boldsymbol\theta}}(x, T)-{u}(x, T) |<\delta_2,
\end{aligned}
\end{equation}
and 
\begin{equation}
    \sup _{(x, T) \in \Omega \times \tau}\left|F_{\Bar{\boldsymbol\theta}}(x, t)-u(x, t)\right| \geq M.
\end{equation}

\end{proof}

The proof of Theorem~\ref{thm:full_pde_constraint} is based on~\citep{peng2020accelerating, gilbarg1977elliptic} by adapting the results on elliptic PDEs to parabolic PDEs. We first give some supporting definitions and lemmas.
We define the second order parabolic operator $L$ w.r.t. $u$ as follows.
\begin{equation}
\label{eq:parabolic_operator}
    L[u]:= -\frac{\partial u}{\partial T} + \sum_{i, j} a_{i, j}(x, T) \frac{\partial^2 u}{\partial x_i x_j}+\sum_i b_i(x, T) \frac{\partial u}{\partial x_i}+c(x, T) u.
\end{equation}
Let $A(x):=\left[a_{i, j}(x, T)\right]_{i, j} \in \mathbb{R}^{d \times d}$ and $b(x, T):=\left[b_i(x)\right]_i \in \mathbb{R}^d$. We consider uniform parabolic operators, where $A$ is positive definite, with the smallest eigenvalue $\lambda > 0$. 

Let $\Omega \times \tau \in \mathbb{R}^{d+1}$ be a bounded domain of interest, where $\Sigma$ is the boundary of $\Omega \times \tau$. We consider the following partial differential equation for the function $u(\cdot)$.
\begin{equation}
\label{eq:parabolic_PDE}
\begin{aligned}
    L[u](x, T) & = q(x, T), \quad (x, T) \in \Omega \times \tau \\
    u(x,T) & = \tilde{u}(x, T), \quad (x, T) \in \Sigma
\end{aligned}
\end{equation}
We know that the risk probability PDE~\eqref{eq:cdc pde} is a parabolic PDE and can be written in the form of~\eqref{eq:parabolic_PDE}.
Specifically, in our case we have $c(x,T) \equiv 0$
and $A = \frac{1}{2}\sigma^2 I$ which gives $\lambda = \frac{1}{2}\sigma^2$.
% and $A$ diagonal, \textit{i.e.}, $a_{i, j} = 0$ if $i \neq j$.

For any space $\Omega \in \mathbb{R}^d$, for any function $f: \mathbb{R}^d \rightarrow \mathbb{R}$, we define the $L_1$ norm of the function $f$ on $\Omega$ to be
\begin{equation}
    \|f\|_{L_1 (\Omega)} := \int_\Omega f(\mathbf{X}) d\mathbf{X}.
\end{equation}
With this definition, we know that 
\begin{equation}   \mathbb{E}_{\mathbf{X}}\left[f\right] =  \|f\|_{L_1 (\Omega)} / |\Omega|,
\end{equation}
where $\mathbf{X}$ is uniformly sampled from $\Omega$, and $|\Omega|$ denote the size of the space $\Omega$.

\begin{corollary}
\label{cor:maximum_principle}
\textbf{Weak maximum principle}. Suppose that $\D = \Omega \times \tau$ is bounded, $L$ is uniformly parabolic with $c \leq 0$ and that $u \in C^0(\bar{\D}) \cap C^2(\D)$ satisfies $Lu \geq 0$ in $\D$, and $M = \max_{\Bar{\D}} u \geq 0$. Then
\begin{equation}
    \max_{\Bar{\D}} u = \max_{\Sigma} u.
\end{equation}
\end{corollary}

\begin{corollary}
\label{cor:comparison_principle}
\textbf{Comparison principle.} Suppose that $\D$ is bounded and $L$ is uniformly parabolic.
If $u, v \in C^0(\bar{\D}) \cap C^2(\D)$ satisfies $Lu \leq Lv$ in $\D$ and $u \geq v$ on $\Sigma$, then $u \geq v$ on $\D$.
\end{corollary}

\begin{theorem}
\label{thm:parabolic_PDE_bound}
Let $Lu = q$ in a bounded domain $\D$, where $L$ is parabolic, $c \leq 0$ and $u \in C^0(\bar{\D}) \cap C^2(\D)$. Then
\begin{equation}
    \sup_\D |u| \leq \sup_{\Sigma} |u| + C \sup_\D \frac{|q|}{\lambda},
\end{equation}
where $C$ is a constant depending only on diam $\D$ and $\beta = \sup |b|/\lambda$.
\end{theorem}

\begin{proof}
(Theorem~\ref{thm:parabolic_PDE_bound})
Let $\D$ lie in the slab $0 < x_1 < d$. Without loss of generality, we assume $\lambda_1 \geq \lambda > 0$. Set $L_0 = -\frac{\partial}{\partial t} + a^{ij}\frac{\partial^2}{\partial x_i x_j} + b^i\frac{\partial}{\partial x_i}$. For $\alpha \geq \beta + 1$, we have
\begin{equation}
    L_0 e^{\alpha x_1} = (\alpha^2 a^{11} + \alpha b^1) e^{\alpha x_1} \geq \lambda (\alpha^2 - \alpha \beta) e^{\alpha x_1} \geq \lambda
\end{equation}
Consider the case when $Lu \geq q$. Let
\begin{equation}
    v = \sup_{\Sigma} u^+ + (e^{\alpha d} - e^{\alpha x_1})\sup_\D \frac{|q^-|}{\lambda},
\end{equation}
where $u^+ = \max (u, 0)$ and $q^- = \min(q, 0)$.
Then, since $Lv = L_0 v + cv \leq -\lambda \sup_{\Sigma} (|q^-|/\lambda)$ by maximum principle (Corollary~\ref{cor:maximum_principle}), we have
\begin{equation}
    L(v-u) \leq -\lambda (\sup_{\D} \frac{|q^-|}{\lambda} + \frac{q}{\lambda}) \leq 0 \text{ in } \D,
\end{equation}
and $v-u \geq 0$ on $\Sigma$. Hence, from comparison principle (Corollary~\ref{cor:comparison_principle}) we have
\begin{equation}
    \sup_{\D} u \leq \sup_{\D} v \leq \sup_{\Sigma} |u| + C \sup_\D \frac{|q|}{\lambda}
\end{equation}
with $C = e^{\alpha d} - 1$.
Consider $Lu \leq q$, we can get similar results with flipped signs. Combine both cases we have for $Lu=q$
\begin{equation}
    \sup_\D |u| \leq \sup_{\Sigma} |u| + C \sup_\D \frac{|q|}{\lambda}.
\end{equation}
\end{proof}

\begin{theorem}
\label{thm:PINN_bound}
Suppose that $\D \in \mathbb{R}^{d+1}$ is a bounded domain, $L$ is the parabolic operator defined in~\eqref{eq:parabolic_operator} and $u \in C^0(\bar{\D}) \cap C^2(\D)$ is a solution to the safety probability PDE. If the PINN ${F}_{\boldsymbol\theta}$ satisfies the following conditions:
\begin{enumerate}
    \item $\sup _{(x, T) \in \Sigma}| {F}_{\boldsymbol\theta}(x, T)-\Tilde{u}(x, T)|<\delta_1$;
    \item $\sup _{(x, T) \in \D}|W_{{F}_{\boldsymbol\theta}}(x,T)|<\delta_2$;
    \item ${F}_{\boldsymbol\theta} \in C^0(\bar{\D}) \cap C^2(\D)$,
\end{enumerate}
Then the error of $F_{\boldsymbol\theta}$ over $\Omega$ is bounded by
\begin{equation}
\label{eq:PINN_bound}
    \sup _{(x, T) \in \D}\left|{F}_{\boldsymbol\theta}(x, T)-u(x, t)\right| \leq {\delta_1}+C \frac{{\delta_2}}{\lambda}
\end{equation}
where $C$ is a constant depending on $\D$ and $L$.
\end{theorem}

\begin{proof}
(Theorem~\ref{thm:PINN_bound})
Denote $h_1 = L\left[{F}_{\boldsymbol\theta}\right] - q$, and $h_2 = {F}_{\boldsymbol\theta} - u$. Since ${F}_{\boldsymbol\theta}$ and $u$ both fall in $C^0(\bar{\D}) \cap C^2(\D)$, from Theorem~\ref{thm:parabolic_PDE_bound} we have
\begin{equation}
\begin{aligned}
    \sup _{\D}\left|h_2(x, t)\right| & \leq \sup_{\Sigma} |h_2(x, t)| + C \sup_\D \frac{|h_1(x, t)|}{\lambda} \\
    & \leq {\delta_1}+C \frac{{\delta_2}}{\lambda} 
\end{aligned}
\end{equation}
which gives~\eqref{eq:PINN_bound}.
\end{proof}

\begin{lemma}
\label{lem:Lipshitz_bound}
% (Lemma 2.3 in~\cite{peng2020accelerating})
Let $\D \subset \mathbb{R}^{d+1}$ be a domain. Define the regularity of $\D$ as

\begin{equation}
    R_{\D}:=\inf _{(x,T) \in \D, r>0} \frac{|B(x, T,  r) \cap \D|}{\min \left\{|\D|, \frac{\pi^{(d+1) / 2} r^{d+1}}{\Gamma((d+1) / 2+1)}\right\}},
\end{equation}
where $B(x, T, r):=\left\{y \in \mathbb{R}^{d+1} \mid\|y-(x,T)\| \leq r\right\}$ and $|S|$ is the Lebesgue measure of a set $S$. Suppose that $\D$ is bounded and $R_{\D}>0$. Let $q \in C^{0}(\bar{\D})$ be an $l_{0}$-Lipschitz continuous function on $\bar{\D}$. Then

\begin{equation}
\label{eq:Lipshitz_bound}
\sup _{\D}|q| \leq \max \left\{\frac{2\|q\|_{L_{1}(\bar{\D})}}{R_{\D}|\D|}, 2 l_{0} \cdot\left(\frac{\|q\|_{L_{1}(\bar{\D})} \cdot \Gamma((d+1) / 2+1)}{l_{0} R_{\D} \cdot \pi^{(d+1) / 2}}\right)^{\frac{1}{d+2}}\right\} .  
\end{equation}
\end{lemma}

\begin{proof}
(Lemma~\ref{lem:Lipshitz_bound})
According to the definition of $l$-Lipschitz continuity, we have
\begin{equation}
    l\|(x,T)-(\bar{x}, \bar{T})\|_2 \geq|q(x,T)-q(\bar{x}, \bar{T})|, \quad \forall (x,T), (\bar{x}, \bar{T}) \in \bar{\D},
\end{equation}
which follows
\begin{equation}
\label{eq:L1_bound}
    \|q\|_{L_{1}(\bar{\D})} \geq \int_{\D^{+}}|q(x, T)| d x dT\geq \int_{\D^{+}}|f(\bar{x}, \bar{T})|-l\|(x,T)-(\bar{x}, \bar{T})\| d x dT,
\end{equation}
where $\D^{+}:=\{(x,T) \in \bar{\D} \mid |q(\bar{x},\bar{T})| -l\|(x,T)-(\bar{x}, \bar{T})\| \geq 0\}$. Without loss of generality, we assume that $(\bar{x}, \bar{T}) \in \arg \max _{\bar{\D}}|q|$ and $q(\bar{x}, \bar{T})>0$. Denote that
\begin{equation}
    B_{1}:=B\left(\bar{x}, \bar{T}, \frac{q(\bar{x}, \bar{T})}{2 l}\right) \cap \D.
\end{equation}
It obvious that $B_{1} \subset \D^{+}$. Note that the Lebesgue measure of a hypersphere in $\mathbb{R}^{d+1}$ with radius $r$ is $\pi^{(d+1) / 2} r^{d+1} / \Gamma((d+1) / 2+1)$. Then~\eqref{eq:L1_bound} becomes
\begin{equation}
\begin{aligned}
    \|q\|_{L_{1}(\bar{\D})} & \geq \int_{B_{1}} q(\bar{x})-l\|(x,T)-(\bar{x}, \bar{T})\| d x dT\\
    & \geq\left|B_{1}\right| \cdot \frac{q(\bar{x}, \bar{T})}{2} \\
    & \geq \frac{q(\bar{x}, \bar{T})}{2} \cdot R_{\D} \cdot \min \left\{|\D|, \frac{\pi^{((d+1) / 2} q(\bar{x}, \bar{T})^{d+1}}{2^{d+1} l^{d+1} \Gamma((d+1) / 2+1)}\right\} \\
    & = \sup_{D} |q| \cdot \frac{R_{\D}}{2} \cdot \min \left\{|\D|, \frac{\pi^{((d+1) / 2} q(\bar{x}, \bar{T})^{d+1}}{2^{d+1} l^{d+1} \Gamma((d+1) / 2+1)}\right\}
\end{aligned}
\end{equation}
which leads to~\eqref{eq:Lipshitz_bound}.
\end{proof} 

Now we are ready to prove Theorem~\ref{thm:full_pde_constraint}.

\begin{proof}
(Theorem~\ref{thm:full_pde_constraint})
From condition 1, $\mathbf{Y}$ is uniformly sampled from $\Sigma$. From condition 2, $\mathbf{X}$ is uniformly sampled from $\D$. Then we have
\begin{equation}
    \mathbb{E}_{\mathbf{Y}}\left[|{F}_{\boldsymbol\theta}(\mathbf{Y})-\Tilde{u}(\mathbf{Y})|\right] = 
    % \mathbb{E}_{\Sigma}\left[|{F}_{\boldsymbol\theta}(x,T)-\Tilde{u}(x,T)|\right] = 
    \| {F}_{\boldsymbol\theta}(x, T)-\Tilde{u}(x, T)\|_{L_1(\Sigma)}/|\Sigma|,
\end{equation}
\begin{equation}
    \mathbb{E}_{\mathbf{X}}\left[|W_{{F}_{\boldsymbol\theta}}(\mathbf{X})|\right] = 
    % \mathbb{E}_{\D}\left[|W_{{F}_{\boldsymbol\theta}}(x,T)|\right] =
    \| W_{{F}_{\boldsymbol\theta}}(x, T)\|_{L_1(\D)}/|\D|.
\end{equation}
From condition 1 and 2 we know that
\begin{equation}
     \| {F}_{\boldsymbol\theta}(x, T)-\Tilde{u}(x, T)\|_{L_1(\Sigma)} < \delta_1 |\Sigma|
\end{equation}
\begin{equation}
      \| W_{{F}_{\boldsymbol\theta}}(x, T)\|_{L_1(\D)} < \delta_2 |\D|
\end{equation}
Also from condition 3 we have that $F_{\boldsymbol\theta}-\tilde{u}$ and $W_{F_{\boldsymbol\theta}}$ are both $l$-Lipschitz continuous on $\bar{\D}$. 
From Lemma~\ref{lem:Lipshitz_bound} we know that
\begin{equation}
\begin{aligned}
    & \quad \sup _{(x, T) \in \Sigma}| {F}_{\boldsymbol\theta}(x, T)-\Tilde{u}(x, T)|  \\
    & \leq \max \left\{\frac{2\| {F}_{\boldsymbol\theta}(x, T)-\Tilde{u}(x, T)\|_{L_1(\Sigma)}}{R_{\Sigma}|\Sigma|}, 2 l \cdot\left(\frac{\| {F}_{\boldsymbol\theta}(x, T)-\Tilde{u}(x, T)\|_{L_1(\Sigma)} \cdot \Gamma(d / 2+1)}{l R_{\Sigma} \cdot \pi^{d / 2}}\right)^{\frac{1}{d+1}}\right\} \\
    & < \max \left\{\frac{2 \delta_1 |\Sigma| }{R_{\Sigma}|\Sigma|}, 2 l \cdot\left(\frac{\delta_1 |\Sigma| \cdot \Gamma(d / 2+1)}{l R_{\Sigma} \cdot \pi^{d / 2}}\right)^{\frac{1}{d+1}}\right\}, \\
    % & = \delta_1 \cdot \max \left\{\frac{2  }{R_{\Sigma}}, 2 l \cdot\left(\frac{ |\Sigma| \cdot \Gamma(d / 2+1)}{l R_{\Sigma} \cdot \pi^{d / 2}}\right)^{\frac{1}{d+1}}\right\}
\end{aligned} 
\end{equation}
\begin{equation}
\begin{aligned}
    & \quad \sup _{(x, T) \in \D}|W_{{F}_{\boldsymbol\theta}}(x,T)| \\
    & \leq \max \left\{\frac{2\| W_{{F}_{\boldsymbol\theta}}(x, T)\|_{L_1(\D)}}{R_{\D}|\D|}, 2 l \cdot\left(\frac{\| W_{{F}_{\boldsymbol\theta}}(x, T)\|_{L_1(\D)} \cdot \Gamma((d+1) / 2+1)}{l R_{\D} \cdot \pi^{(d+1) / 2}}\right)^{\frac{1}{d+2}}\right\} \\
    & < \max \left\{\frac{2\delta_2 |\D|}{R_{\D}|\D|}, 2 l \cdot\left(\frac{\delta_2 |\D| \cdot \Gamma((d+1) / 2+1)}{l R_{\D} \cdot \pi^{(d+1) / 2}}\right)^{\frac{1}{d+2}}\right\}. \\
    % & = \delta_2 \cdot \max \left\{\frac{2}{R_{\D}}, 2 l \cdot\left(\frac{|\D| \cdot \Gamma((d+1) / 2+1)}{l R_{\D} \cdot \pi^{(d+1) / 2}}\right)^{\frac{1}{d+2}}\right\}.
\end{aligned}
\end{equation}
Let
% \begin{equation}
% \begin{aligned}
%     C_1 & = \max \left\{\frac{2  }{R_{\Sigma}}, 2 l \cdot\left(\frac{ |\Sigma| \cdot \Gamma(d / 2+1)}{l R_{\Sigma} \cdot \pi^{d / 2}}\right)^{\frac{1}{d+1}}\right\}, \\
%     C_2 & = \max \left\{\frac{2}{R_{\D}}, 2 l \cdot\left(\frac{|\D| \cdot \Gamma((d+1) / 2+1)}{l R_{\D} \cdot \pi^{(d+1) / 2}}\right)^{\frac{1}{d+2}}\right\}.
% \end{aligned}
% \end{equation}
\begin{equation}
\begin{aligned}
    \tilde \delta_1 & = \max \left\{\frac{2 \delta_1 |\Sigma| }{R_{\Sigma}|\Sigma|}, 2 l \cdot\left(\frac{\delta_1 |\Sigma| \cdot \Gamma(d / 2+1)}{l R_{\Sigma} \cdot \pi^{d / 2}}\right)^{\frac{1}{d+1}}\right\}, \\
    \tilde \delta_2 & = \max \left\{\frac{2\delta_2 |\D|}{R_{\D}|\D|}, 2 l \cdot\left(\frac{\delta_2 |\D| \cdot \Gamma((d+1) / 2+1)}{l R_{\D} \cdot \pi^{(d+1) / 2}}\right)^{\frac{1}{d+2}}\right\}.
\end{aligned}
\end{equation}
Then from Theorem~\ref{thm:PINN_bound} we know that 
% \begin{equation}
%     \sup _{(x,T) \in \D}\left|{F}_{\boldsymbol\theta}(x, T)-u(x, T)\right| \leq C_1 \delta_1 + C_2 \frac{\delta_2}{\lambda},
% \end{equation}
\begin{equation}
    \sup _{(x,T) \in \D}\left|{F}_{\boldsymbol\theta}(x, T)-u(x, T)\right| \leq \tilde\delta_1 + C \frac{\tilde \delta_2}{\lambda}.
\end{equation}
Given that $\lambda = \frac{1}{2}\sigma^2$, replace $C$ with $2C$ we get
\begin{equation}
    \sup _{(x,T) \in \D}\left|{F}_{\boldsymbol\theta}(x, T)-u(x, T)\right| \leq \tilde\delta_1 + C \frac{\tilde \delta_2}{\sigma^2},
\end{equation}
which completes the proof.

\end{proof}

\section{Additional Theorems}
\begin{theorem}
\label{thm:PINN_worst_case_bound}
Suppose that $\D \in \mathbb{R}^{d+1}$ is a bounded domain and $u \in C^0(\bar{\D}) \cap C^2(\D)$ is a solution to the PDE of interest. 
Let $\Sigma_s$ be a strict sub-region in $\Sigma$, and $\D_s$ be a sub-region in $\D$.
Given the PINN ${F}_{\boldsymbol\theta}$ parameterized by $\boldsymbol\theta$ satisfies the following conditions: 
% \vspace{-0.02in}
\begin{enumerate}
    \item 
    $\sup _{(x, T) \in  \Sigma_s}|{F}_{\boldsymbol\theta}(x, T)-u(x, T)|<\delta_1$
    \item 
    $\sup _{(x, T) \in \D_s}|W_{{F}_{\boldsymbol\theta}}(x, T)|<\delta_2$
    \item 
    $\sup _{(x, T) \in \D_s}|{F}_{\boldsymbol\theta}(x, T)-{u}(x, T) |<\delta_3$
    % \item 
    % ${F}_{\boldsymbol\theta} \in C^0(\bar{\D}) \cap C^2(\D)$
\end{enumerate}
% \vspace{-0.02in}
With ${F}_{\boldsymbol\theta}$ has sufficient representation, $\forall M > 0$, there exist an trained instance ${F}_{\Bar{\boldsymbol\theta}}$ such that 
\begin{equation}
\label{eq:PINN_worst_case_bound}
    \sup _{(x, T) \in \D}\left|{F}_{\Bar{\boldsymbol\theta}}(x, T)-u(x, T)\right| \geq M.
\end{equation}
\end{theorem}

\begin{proof}(Theorem~\ref{thm:PINN_worst_case_bound})
We know that $u$ is the solution to the PDE of interest. 
% Let $u^* = u$ be the solution on $\D_s$. 
We can construct $\Bar{u}$ such that
\begin{equation}
    \bar{u} = \begin{cases}
    u, \quad (x,T) \in \D_s \\
    u + \frac{d}{d_{\text{max}}}(M+\delta), \quad (x,T) \in \Omega \times \tau \backslash \D_s
    \end{cases}
\end{equation}
where $d$ characterizes the distance between $(x,t) \in \Omega \backslash \D_s$ and set $\D_s$, $d_{\text{max}} = \max_{(x,t) \in \Omega \backslash \D_s} d$ is the maximum distance, and $0 < \delta < \min\{\delta_1, \delta_3\}$ is a constant. Then we have
\begin{equation}
    \sup_{(x, T) \in \Omega \times \tau} \left|\bar{u}(x, t)-u(x, t)\right| = M + \delta
\end{equation}
Since we assume the neural network has sufficient representation power, by universal approximation theorem~\citep{cybenko1989approximation}, for $\delta$ given above, there exist $\Bar{\boldsymbol\theta}$ such that $\sup _{(x, T) \in \D_s}|W_{F_{\Bar{\boldsymbol\theta}}}(x, T)|<\delta_2$ and
\begin{equation}
    \sup _{(x, T) \in \Omega \times \tau}\left|F_{\Bar{\boldsymbol\theta}}(x, t)-\bar{u}(x, t)\right| \leq \delta.
\end{equation}
Then we have $\Bar{\boldsymbol\theta}$ satisfies the three conditions, and 
\begin{equation}
    \sup _{(x, T) \in \Omega \times \tau}\left|F_{\Bar{\boldsymbol\theta}}(x, t)-u(x, t)\right| \geq M.
\end{equation}

\end{proof}

Theorem~\ref{thm:PINN_worst_case_bound} says that when we only impose data and PDE constraints on a sub-region in the space-time domain, the learned PINN can perform arbitrarily inaccurate due to insufficient physical constraints. One special case of this setting is $\mathcal{P} = \emptyset$, which is the typical configuration for standard neural networks without physical constraints (Corollary~\ref{cor:NN_worst_case_bound}).
We also point out that Theorem~\ref{thm:PINN_worst_case_bound} is a worst-case result. In practice, one may observe that the PINN performs well on the full space-time domain even with constraints imposed on sub-region.

\section{Experiment Details}
In this section, we provide details for the four experiments presented in the paper.

\subsection{Generalization to unseen regions}
In the generalization task in section~\ref{sec:generalization}, we use a down-sampled sub-region of the system to train the proposed PIPE framework, and test the prediction results on the full state-time space. We showed that PIPE is able to give accurate inference on the entire state space, while standard fitting strategies cannot make accurate predictions. In the paper, we only show the fitting results of thin plate spline interpolation. Here, we show the results of all fitting strategies we tested for this generalization tasks. The fitting strategies are
\begin{enumerate}
    \item Polynomial fitting of 5 degrees for both state $x$ and time $T$ axes. The fitting sum of squares error (SSE) on the training data is $0.1803$.
    \item Lowess: locally weighted scatterplot smoothing~\citep{cleveland1981lowess}. The training SSE is $0.0205$.
    \item Cubic spline interpolation. The training SSE is $0$.
    \item Biharmonic spline interpolation. The training SSE is $2.52\times 10^{-27}$.
    \item TPS: thin plate spline interpolation. The training SSE is $1.64 \times 10^{-26}$.
\end{enumerate}
All fittings are conducted via the MATLAB Curve Fitting Toolbox.
Fig.~\ref{fig:fitting} visualizes the fitting results on the full state space. Polynomial fitting performs undesirably because the polynomial functions cannot capture the risk probability geometry well. Lowess fitting also fails at inference since it does not have any model information of the data. Given the risk probability data, cubic spline cannot extrapolate outside the training region, and we use $0$ value to fill in the unseen region where it yields NAN for prediction. Biharmonic and TPS give similar results as they are both spline interpolation methods. None of these fitting methods can accurately predict the risk probability in unseen regions, because they purely rely on training data and do not incorporate any model information of the risk probability for prediction.

\begin{figure}[h]
    \centering
    \includegraphics[width=11cm]{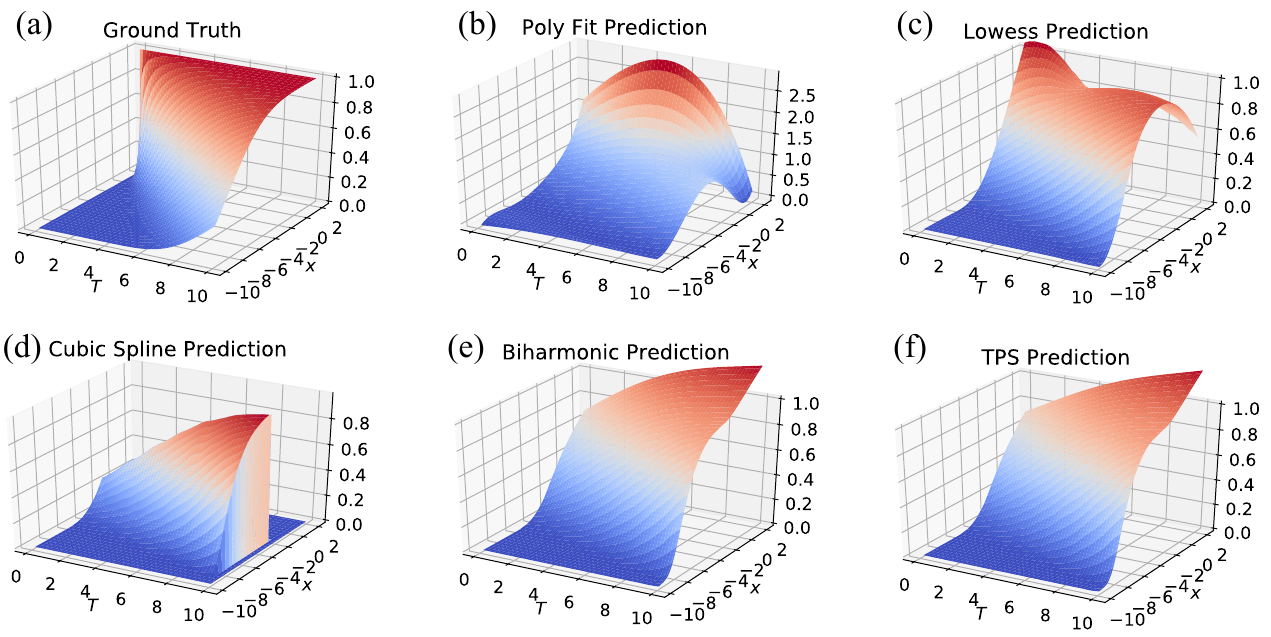}
    \caption{Results of different fitting strategies on the risk probability generalization task.}
    \label{fig:fitting}
\end{figure}

% \todo{Generalization of different parameters}
We also compare the prediction results for different network architectures in the proposed PIPE framework, to examine the effect of network architectures on the risk probability prediction performance.
The network settings we consider are different hidden layer numbers (1-4) and different numbers of neurons in one hidden layer (16, 32, 64). We use 3 hidden layers, 32 neurons per layer as baseline (the one used in the paper). Table~\ref{tb:prediction layer number} and Table~\ref{tb:prediction neuron number} report the averaged absolute error of the predictions for different layer numbers and neuron numbers per layer, respectively. We trained the neural networks 10 times with random initialization to report the standard deviation.
We can see that as the number of layers increases, the prediction error of the risk probability drops, but in a relatively graceful manner. The prediction error for a single layer PIPE is already very small, which indicates that the proposed PIPE framework is very efficient in terms of computation and storage. The prediction accuracy tends to saturate when the hidden layer number reaches 3, as there is no obvious improvement when we increase the layer number from 3 to 4. This means for the specific task we consider, a 3-layer neural net has enough representation.
Under the same layer number, as the neuron number per layer increases, the risk probability prediction error decreases. This indicates that with larger number of neuron in each layer (\ie wider neural networks), the neural network can better capture the mapping between state-time pair and the risk probability. However, the training time increases significantly for PIPEs with more neurons per layer ($152\mathrm{s}$ for 16 neurons and $971\mathrm{s}$ for 64 neurons), and the gain in prediction accuracy becomes marginal compared to the amount of additional computation involved. We suggest to use a moderate number of neurons per layer to achieve desirable trade-offs between computation and accuracy.

\begin{table*}[h]
\centering
%\resizebox{.95\columnwidth}{!}{
\begin{tabular}{l|cccc}
    \hline
    \textbf{\# Hidden Layer} & 1 & 2 & 3 & 4 \\
    \hline
    \textbf{Prediction Error ($\times 10^{-3}$)} & $4.773\pm 0.564$ & $\boldsymbol{2.717 \pm 0.241}$ & $2.819 \pm 0.619$ & $2.778 \pm 0.523$ \\
    \hline
\end{tabular}
\caption{Risk probability prediction error of PIPE for different numbers of hidden layers.}
\label{tb:prediction layer number}
\end{table*}

\begin{table*}[h]
\centering
%\resizebox{.95\columnwidth}{!}{
\begin{tabular}{l|ccc}
    \hline
    \textbf{\# Neurons} & 16 & 32 & 64 \\
    \hline
    \textbf{Prediction Error ($\times 10^{-3}$)} & $2.743 \pm 0.313$ & $2.931 \pm 0.865$ & $\boldsymbol{2.599 \pm 0.351}$ \\
    \hline
\end{tabular}
\caption{Risk probability prediction error of PIPE for different neuron numbers per layer.}
\label{tb:prediction neuron number}
\end{table*}

\subsection{Efficient estimation of risk probability}
% \todo{visualization of estimation plots}
In the efficient estimation task in section~\ref{sec:estimation}, we showed that PIPE will give better sample efficiency in risk probability prediction, in the sense that it achieves the same prediction accuracy with less sample numbers. Here, we visualize the prediction errors of Monte Carlo (MC) and the proposed PIPE framework to better show the results. Fig.~\ref{fig:comparison} shows the prediction error comparison plots for MC and PIPE with different sample numbers $N$. As the sample number increases, the errors for both MC and PIPE decrease because of better resolution of the data. PIPE gives 
more accurate predictions than MC across all sample numbers, since it combines data and physical model of the system together. From Fig.~\ref{fig:comparison} we can see that, PIPE indeed provides smoother and more accurate estimation of the risk probability. The visualization results further validate the efficacy of the proposed PIPE framework.
\begin{figure}
    \centering
    \includegraphics[width=15cm]{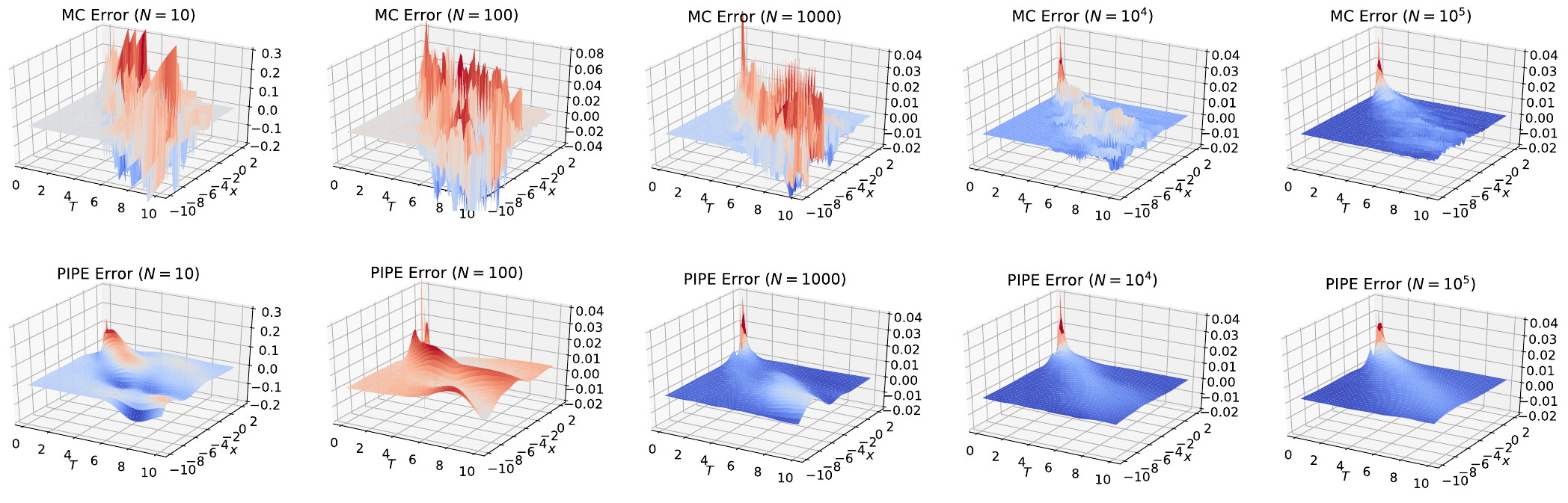}
    \caption{Prediction errors for Monte Carlo (left) and PIPE (right) with different sample numbers.}
    \label{fig:comparison}
\end{figure}

\begin{table*}[h]
\centering
%\resizebox{.95\columnwidth}{!}{
\begin{tabular}{l|cccc}
    \hline
    \textbf{\# Hidden Layer} & 1 & 2 & 3 & 4 \\
    \hline
    \textbf{Prediction Error ($\times 10^{-4}$)} & $14.594 \pm 2.109$ & $7.302 \pm 0.819$ & $6.890 \pm 0.613$ & $\boldsymbol{6.625 \pm 0.574}$ \\
    \hline
\end{tabular}
\caption{Risk probability gradient prediction error of PIPE for different numbers of hidden layers.}
\label{tb:gradient layer number}
\end{table*}

\begin{table*}
\centering
%\resizebox{.95\columnwidth}{!}{
\begin{tabular}{l|ccc}
    \hline
    \textbf{\# Neurons} & 16 & 32 & 64 \\
    \hline
    \textbf{Prediction Error ($\times 10^{-4}$)} & $7.049 \pm 0.767$ & $6.890 \pm 0.613$ & $\boldsymbol{6.458 \pm 0.794}$ \\
    \hline
\end{tabular}
\caption{Risk probability gradient prediction error of PIPE for different neuron numbers per layer.}
\label{tb:gradient neuron number}
\end{table*}

\subsection{Adaptation on changing system parameters}
% \todo{adaptation of more parameters}
For the adaptation task described in section~\ref{sec:adaptation}, we trained PIPE with system data of parameters $\lambda_\text{train} = [0.1,0.5,0.8,1]$ and tested over a range of unseen parameters over the interval $\lambda=[0,2]$. Here, we show additional results on parameters $\lambda_\text{test} = [0.3, 0.7, 1.2, 2]$ to further illustrate the adaptation ability of PIPE. Fig.~\ref{fig:varying-new} shows the results. It can be seen that PIPE is able to predict the risk probability accurately on both system parameters with very low error over the entire state-time space. This result indicates that PIPE has solid adaptation ability on uncertain parameters, and can be used for stochastic safe control with adaptation requirements.

\begin{figure}
    \centering
    \includegraphics[width=12cm]{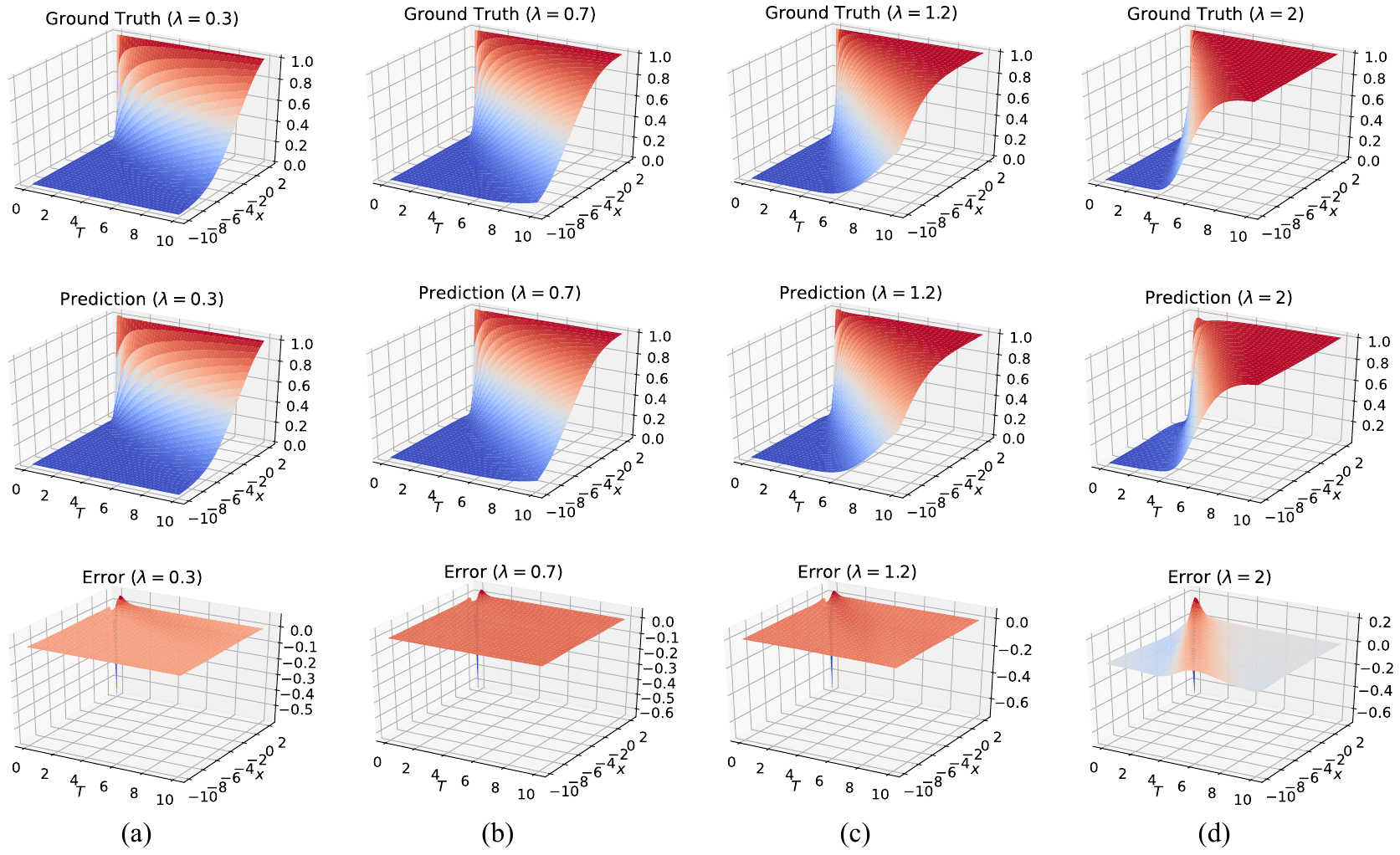}
    \caption{Risk probability prediction results on systems with unseen parameters $\lambda_\text{test} = [0.3, 0.7, 1.2, 2]$.}
    \label{fig:varying-new}
\end{figure}

\subsection{Estimating the gradient of risk probability}
% \todo{Generalization of different parameters}
For the gradient estimation task in section~\ref{sec:gradient}, we presented that PIPE is able to predict the risk probability gradients accurately by taking finite difference on the predicted risk probabilities. This result shows that PIPE can be used for first- and higher-order methods for safe control, by providing accurate gradient estimations in a real-time fashion. Similar to the generalization task, here we report the gradient prediction errors with different network architectures in PIPE, to examine the effect of network architectures on the gradient estimation performance. Table~\ref{tb:gradient layer number} and Table~\ref{tb:gradient neuron number} show the averaged absolute error of gradient predictions for different layer numbers and neuron numbers per layer. We trained the neural networks for 10 times with random initialization to report the standard deviation. 
We can see that as the number of hidden layer increases, the gradient prediction error keeps dropping, and tends to saturate after 3 layers.  
With the increasing neuron numbers per layer, the gradient prediction error decreases in a graceful manner. Similar to the generalization task, larger networks with more hidden layers and more neurons per layer can give more accurate estimation of the gradient, but the computation time scales poorly compared to the accuracy gain. Based on these results, we suggest to use moderate numbers of layers and neurons per layer to acquire desirable gradient prediction with less computation time.

\section{Additional Experiment Results}
In this section, we provide additional experiment results on risk estimation of the inverted pendulum on a cart system, as well as safe control with risk estimation through the proposed PIPE framework.

\subsection{Risk estimation of inverted pendulum on a cart system}
We consider the inverted pendulum on a cart system, with dynamics given by
\begin{equation}
\label{eq:pendulum_dynamics}
    \frac{d\mathbf{x}}{dt} = f(\mathbf{x}) + g(\mathbf{x}) u + \sigma \ \Tilde{I} \ dW_t,
\end{equation}
where $\mathbf{x} = [x, \dot x, \theta, \dot \theta]^\top$ is the state of the system and $u \in \mathbb{R}$ is the control, with $x$ and $\dot x$ being the position and velocity of the cart, and $\theta$ and $\dot \theta$ being the angle and angular velocity of the pendulum. We use $\mathbf{x}$ to denote the state of the system to distinguish from cart's position $x$. Then, we have
\begin{equation}
f(\mathbf{x}) = \left[\begin{array}{cccc}
1 & 0 & 0 & 0 \\
0 & m+M & 0 & m l \cos \theta \\
0 & 0 & 1 & 0 \\
0 & m l \cos \theta & 0 & m l^2 
\end{array}\right]^{-1} \left[\begin{array}{c}
\dot{x} \\
m l \dot{\theta}^2 \sin \theta-b_x \dot{x} \\
\dot{\theta} \\
m g l \sin \theta-b_\theta l \dot{\theta}
\end{array}\right],
\end{equation}
\begin{equation}
g(\mathbf{x}) = \left[\begin{array}{cccc}
1 & 0 & 0 & 0 \\
0 & m+M & 0 & m l \cos \theta \\
0 & 0 & 1 & 0 \\
0 & m l \cos \theta & 0 & m l^2 
\end{array}\right]^{-1} \left[\begin{array}{l}
0 \\
1 \\
0 \\
0
\end{array}\right],
\end{equation}
where $m$ and $M$ are the mass of the pendulum and the cart, $g$ is the acceleration of gravity, $l$ is the length of the pendulum, and $b_x$ and $b_\theta$ are constants. The last term in~\eqref{eq:pendulum_dynamics} is the additive noise, where $W_t$ is $4$-dimensional Wiener process with $W_0 = \mathbf{0}$, $\sigma$ is the magnitude of the noise, and
\begin{equation}
\Tilde{I} = \left[\begin{array}{cccc}
0 & 0 & 0 & 0 \\
0 & 1 & 0 & 0 \\
0 & 0 & 0 & 0 \\
0 & 0 & 0 & 1
\end{array}\right].
\end{equation}
Fig.~\ref{fig:cartpend} visualizes the system.
% \begin{equation}
% \left[\begin{array}{cccc}
% 1 & 0 & 0 & 0 \\
% 0 & m+M & 0 & m l \cos \theta \\
% 0 & 0 & 1 & 0 \\
% 0 & m l \cos \theta & 0 & m l^2
% \end{array}\right] \frac{d}{d t}\left[\begin{array}{c}
% x \\
% \dot{x} \\
% \theta \\
% \dot{\theta}
% \end{array}\right]=\left[\begin{array}{c}
% \dot{x} \\
% m l \dot{\theta}^2 \sin \theta-b_x \dot{x} \\
% \dot{\theta} \\
% m g l \sin \theta-b_\theta l \dot{\theta}
% \end{array}\right]+\left[\begin{array}{l}
% 0 \\
% 1 \\
% 0 \\
% 0
% \end{array}\right] U + 
% \left[\begin{array}{cccc}
% 0 & 0 & 0 & 0 \\
% 0 & \sigma & 0 & 0 \\
% 0 & 0 & 0 & 0 \\
% 0 & 0 & 0 & \sigma
% \end{array}\right] dW_t
% \end{equation}
% where $x$ and $\dot x$ are the cart's position and velocity, $\theta$ and $\dot \theta$ are the angle and angular velocity of the pendulum, $m$ and $M$ are the mass of the pendulum and the cart, $l$ is the length of the pendulum, and $b_x$ and $b_\theta$ are constants. We use $\mathbf{x}$ to denote the state of the system to distinguish from cart's position $x$.

The safe set is defined in~\eqref{eq:safe set definition} with barrier function $\phi(\mathbf{x}) = 1-(\frac{\mathbf{x}_3}{\pi/3})^2 = 1-(\frac{\theta}{\pi/3})^2$. Essentially the system is safe when the angle of the pendulum is within $[-\pi/3, \pi/3]$.
\begin{figure}
\includegraphics[width=6cm]{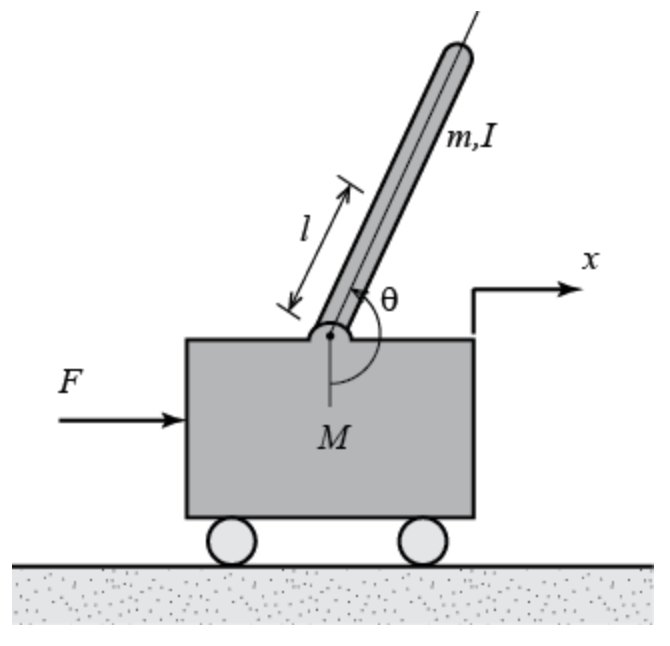}
\centering
\caption{Inverted pendulum on a cart.}
\label{fig:cartpend}
\end{figure}
We consider spatial-temporal space $\Omega \times \tau = \left[[-10,10] \times [-\pi/3, \pi/3] \times [-\pi, \pi] \right] \times [0,1]$. We collect training and testing data on $\Omega \times \tau$ with grid size $N_{\Omega\text{-train}}=13$ and $N_{\tau\text{-train}} = 10$ for training and $N_{\Omega\text{-test}}=25$ and $N_{\tau\text{-test}} = 10$ for testing. The nominal controller we choose is a proportional controller $N(\mathbf{x}) = -K \mathbf{x}$ with $K = [0, -0.9148, -22.1636, -14.3992]^\top$.
The sample number for MC simulation is set to be $N = 1000$ for both training and testing.
% We set $\sigma = 8$ to be fixed.  
Table~\ref{tb:pendulum_parameters} lists the parameters used in the simulation.

We train PIPE with the same configuration listed in section~\ref{sec:experiments}. 
According to Theorem~\ref{thm:cdc pde}, the PDE that characterizes the safety probability of the pendulum system is
\begin{equation}
\label{eq:pendulum_pde}
\begin{aligned}
    \frac{\partial F}{\partial T}(\mathbf{x}, T) 
    & = \left(f(\mathbf{x}) - g(\mathbf{x}) K \mathbf{x} \right) \frac{\partial F}{\partial \mathbf{x}}(\mathbf{x}, T) + \frac{1}{2} \sigma^2 \Tilde{I} \operatorname{tr}\left( \frac{\partial^{2} F}{\partial \mathbf{x}^{2}}(\mathbf{x}, T) \right),
\end{aligned}
\end{equation}
which is a high dimensional and highly nonlinear PDE that cannot be solved effectively using standard solvers. Here we can see the advantage of combining data with model and using a learning-based framework to estimate the safety probability.
Fig.~\ref{fig:zoomin}, Fig.~\ref{fig:pendulum_t03} and Fig.~\ref{fig:pendulum_t1} show the results of the PIPE predictions. We see that despite the rather high dimension and nonlinear dynamics of the pendulum system, PIPE is able to predict the safety probability of the system with high accuracy. Besides, since PIPE takes the model knowledge into training loss, the resulting safety probability prediction is smoother thus more reliable than pure MC estimations.

\begin{table*}
\centering
%\resizebox{.95\columnwidth}{!}{
\begin{tabular}{|c|c|}
    \hline
    Parameters & Values \\
    \hline
    $M$ & $1$ \\
    $m$ & $0.1$ \\
    $g$ & $9.8$ \\
    $l$ & $0.5$ \\
    $b_x$ & $0.05$ \\
    $b_\theta$ & $0.1$ \\
    \hline
\end{tabular}
\caption{Parameters used in the inverted pendulum simulation.}
\label{tb:pendulum_parameters}
\end{table*}

\begin{figure}
\includegraphics[width=15cm]{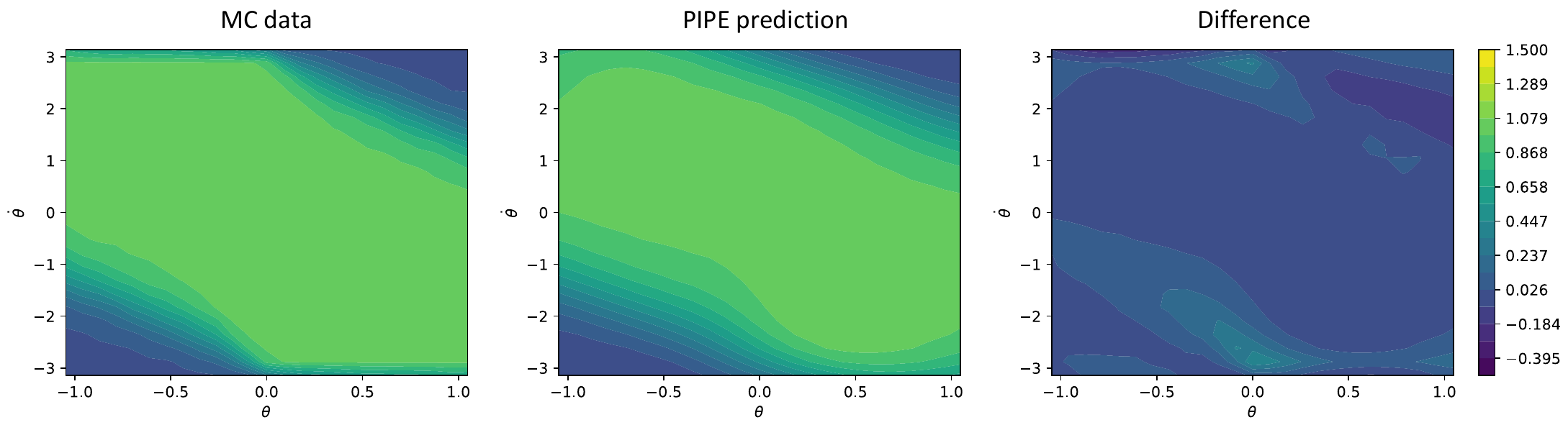}
\centering
\caption{Safety probability from MC simulation and PIPE prediction, and their difference. Results on outlook time horizon $T = 0.6$, initial velocity $v = 0$. The x-axis shows the initial angle, and y-axis shows the initial angular velocity.}
\label{fig:zoomin}
\end{figure}

\begin{figure}
\includegraphics[width=16cm]{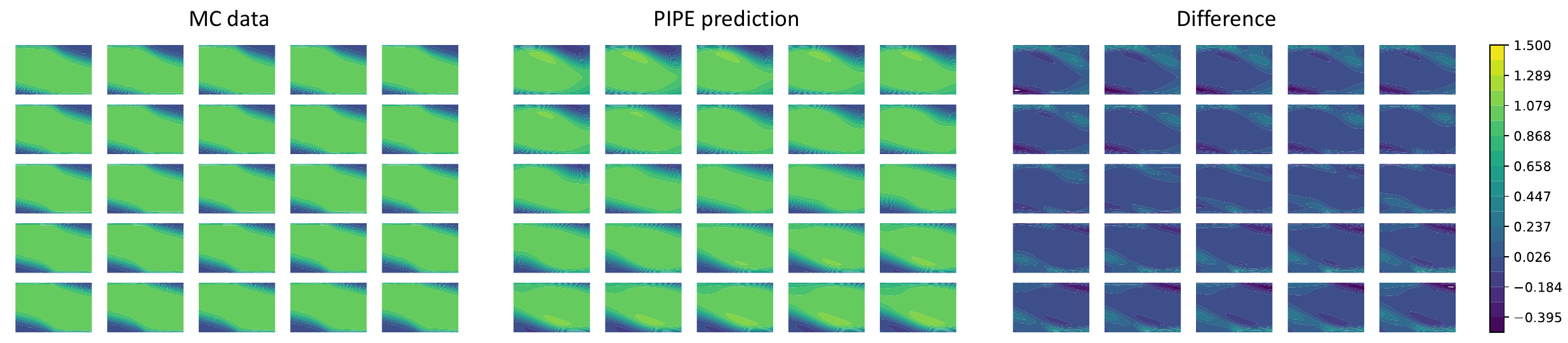}
\centering
\caption{Safety probability from MC simulation and PIPE prediction, and their difference. Results on outlook time horizon $T = 0.3$. The 5x5 plots show results on 25 different initial velocities uniformly sampled in $[-10, 10]$. The x-axis and y-axis (omitted) are the initial angle and the initial angular velocity as in Fig.~\ref{fig:zoomin}. One can see that the safety probability shift as the velocity changes, and the safety probability is symmetric with regard to the origin $v=0$ due to the symmetry of the dynamics.}
\label{fig:pendulum_t03}
\end{figure}

\begin{figure}
\includegraphics[width=16cm]{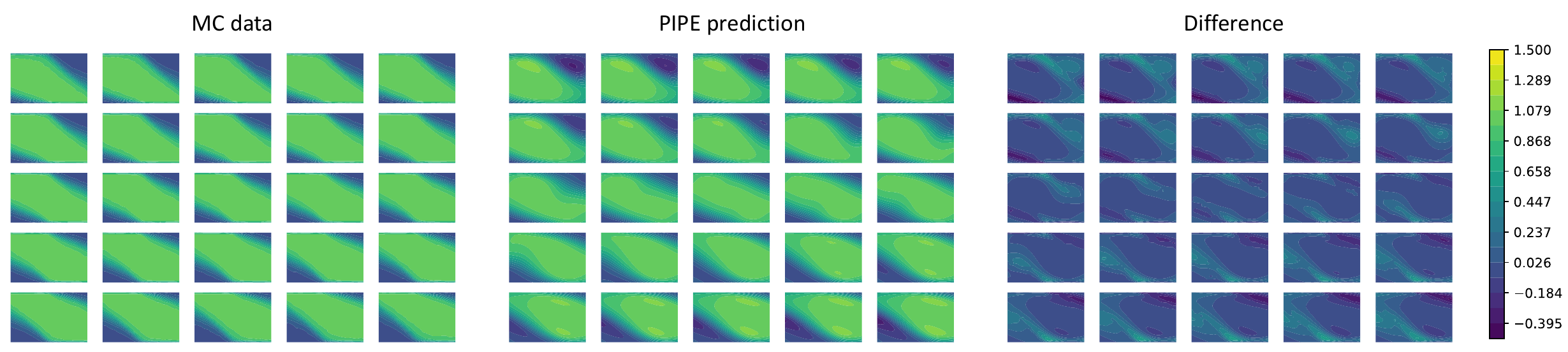}
\centering
\caption{Safety probability from MC simulation and PIPE prediction, and their difference. Results on outlook time horizon $T = 1$. The 5x5 plots show results on 25 different initial velocities uniformly sampled in $[-10, 10]$.}
\label{fig:pendulum_t1}
\end{figure}

\subsection{Safe control with PIPE}
We consider the control affine system~\eqref{eq:cdc system} with $f(x) \equiv Ax = 2x$, $g(x) \equiv 1$, $\sigma(x) \equiv 2$. %\jacob{looked at Clark's paper, should we change it into $\sigma = 2I$?} 
The safe set is defined as in~\eqref{eq:safe set definition} and the barrier function is chosen to be $\phi(x) := x-1$. The safety specification is given as the forward invariance condition. 
The nominal controller is a proportional controller $N(x) = -K x$ with $K = 2.5$. The closed-loop system with this controller has an equilibrium at $x=0$ and tends to move into the unsafe set in the state space. We run simulations with $d t = 0.1$ for all controllers. The initial state is set to be $x_0 = 3$. 
For this system, the safety probability satisfies the following PDE
\begin{equation}
\label{eq:safety pde}
\begin{aligned}
    \frac{\partial F}{\partial T}(x, T) 
    & = -0.5x \frac{\partial F}{\partial x}(x, T) + 2 \operatorname{tr}\left( \frac{\partial^{2} F}{\partial x^{2}}(x, T) \right),
\end{aligned}
\end{equation}
with initial and boundary conditions
\begin{equation}
\begin{aligned}
\label{eq:icbc}
F(x,t) & = 0, \; x \leq 1, \\
F(x,0) & = \mathbbm{1}(x \geq 1).
\end{aligned}
\end{equation}

% \todo{PIPE estimation}
We first estimate the safety probability $F(x,T)$ of the system via PIPE. 
% The state-time region of interest is $\Omega \times \tau = [1,10] \times [0,10]$.
The training data $\bar{F}(x,T)$ is acquired by running MC on the system dynamics for given initial state $x_0$ and nominal control $N$. Specifically,
\begin{equation}
    \bar{F}(x,T) = \pr(\forall t \in [0,T], x_t \in \mathcal{C} \mid x_0 = x) = \frac{N_\text{safe}}{N},
\end{equation}
with $N=100$ being the number of sample trajectories. The training data is sampled on the state-time region $\Omega \times \tau = [1,10] \times [0,10]$ with grid size $dx = 0.5$ and $dt = 0.5$. 
% We train the PINN with Adam optimizer with learning rate 0.001, and then
% We use PINN with 3 hidden layers and 32 neurons per layer. The activation function is chosen as hyperbolic tangent function ($\tanh$). We use Adam optimizer~\citep{kingma2014adam} for training with initial learning rate set as $0.001$. The PINN parameters $\boldsymbol\theta$ is initialized via Glorot uniform initialization. The weights in the loss function~\eqref{eq:PINN overall loss function} are set to be $\omega_p = \omega_d = 1$. We train the PINN for 60000 epochs in this experiment.
We train PIPE with the same configuration as listed in section~\ref{sec:experiments}.
We test the estimated safety probability and its gradient on the full state space $\Omega \times \tau$ with $dx = 0.1$ and $dt = 0.1$. Fig.~\ref{fig:PIPE estimation} shows the results. It can be seen that the PIPE estimate is very close to the Monte Carlo samples, which validates the efficacy of the framework. Furthermore, the PIPE estimation has smoother gradients, due to the fact that it leverages model information along with the data.
\begin{figure}
    \centering
    \includegraphics[width=11cm]{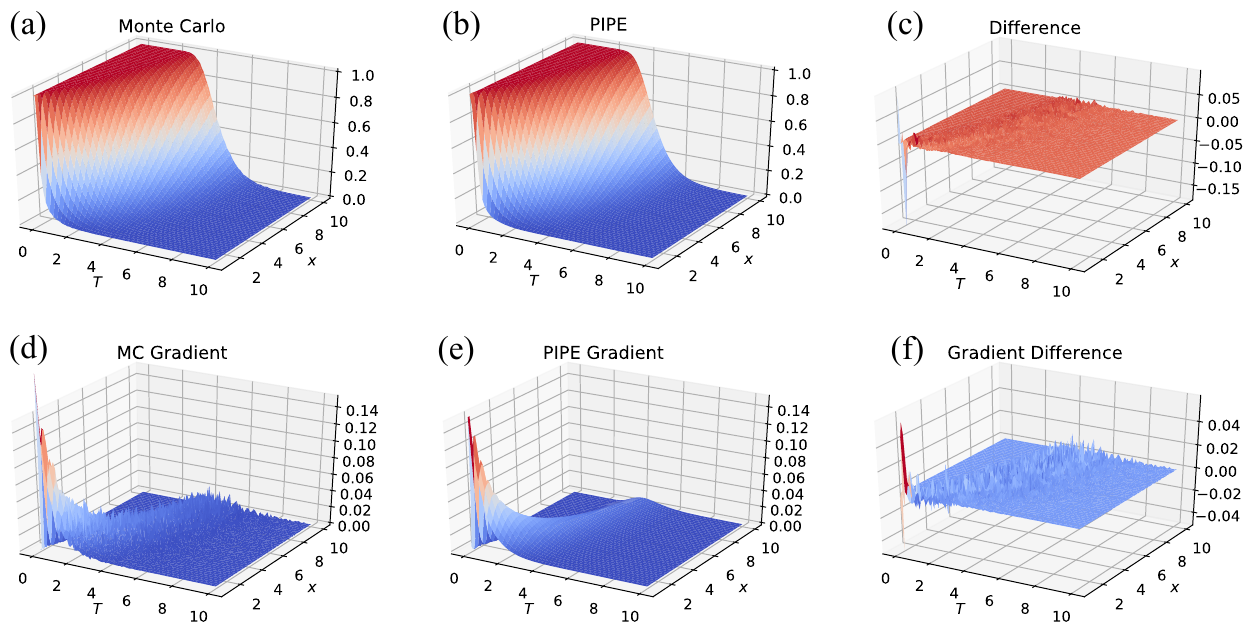}
    \caption{Safety probability and its gradient of Monte Carlo samples and PIPE estimation.}
    \label{fig:PIPE estimation}
\end{figure}
% \todo{prediction and error for probability and gradient}

% \todo{safe control with PIPE estimation}
We then show the results of using such estimated safety probability for control. Fig.~\ref{fig:PIPE control} shows the results. For the baseline stochastic safe controllers for comparison, refer to~\cite{wang2021myopically} for details. We see that the PIPE framework can be applied to long-term safe control methods discussed in~\citep{wang2021myopically}. Together with the PIPE estimation, long-term safety can be ensured in stochastic control systems.
\begin{figure}
    \centering
    \includegraphics[width=12cm]{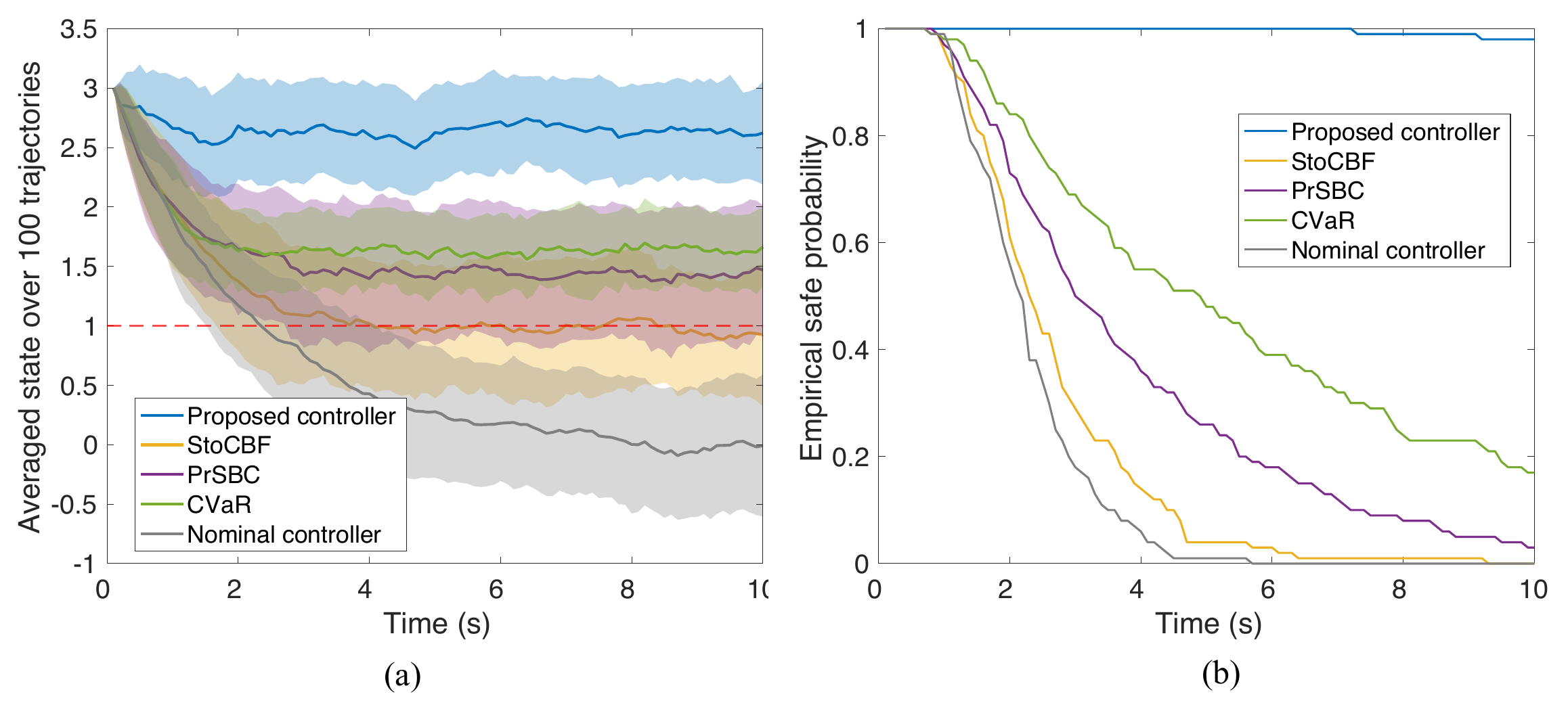}
    \caption{Safe control with probability estimation from PIPE compared with other baselines.}
    \label{fig:PIPE control}
\end{figure}

\end{document}